\documentclass[format=acmsmall, review=false, screen=true]{acmart}

\usepackage{booktabs} 
\usepackage{pdflscape} 
\usepackage{afterpage} 

\usepackage[ruled]{algorithm2e} 

\SetAlFnt{\small}
\SetAlCapFnt{\small}
\SetAlCapNameFnt{\small}
\SetAlCapHSkip{0pt}
\IncMargin{-\parindent}

\renewcommand\footnotetextcopyrightpermission[1]{} 
\usepackage{inputenc}
\usepackage{balance}
\usepackage{xspace}
\usepackage{rotating}
\usepackage{url}
\usepackage{paralist}
\usepackage{fixltx2e}
\usepackage{multirow}
\usepackage{amsmath,amssymb,amsthm}
\usepackage{xfrac}
\usepackage{array}
\usepackage{textcomp}
\usepackage{siunitx}
\usepackage{pgfplots}
\usepackage{xcolor}
\pgfplotsset{compat=newest,unit code/.code={\si{#1}},plot coordinates/math parser=false,grid style={lightgray}}
\usepgfplotslibrary{units,external,groupplots}
\usetikzlibrary{positioning,angles,quotes,patterns}

\newtheorem{theo}{Theorem}
\newtheorem{lem}{Lemma}
\newtheorem{defi}{Definition}

\DeclareMathOperator*{\E}{\mathbb{E}}
\DeclareMathOperator*{\Var}{Var}

\DeclareMathOperator*{\R}{\mathbb{R}}

\newcommand\figref[1]{Fig.~\ref{#1}}
\newcommand\tabref[1]{Table~\ref{#1}}
\newcommand\secref[1]{Section~\ref{#1}}

\newcommand{\etal}{et~al.\xspace}
\newcommand{\eg}{\emph{e.g.},\xspace}
\newcommand{\ie}{\emph{i.e.},\xspace}
\newcommand{\cf}{cf.\xspace}
\newcommand{\etc}{etc.\xspace}
\newcommand{\capt}[1]{\mdseries{\emph{#1}}}

\newcommand{\iid}{i.i.d.\xspace}

\usepackage{ifthen}
\newboolean{authnotes}

\setboolean{authnotes}{false}

\ifthenelse{\boolean{authnotes}}
{
\newcommand{\fm}[1]{\footnote{{\bf\color{blue} Fabian: #1}}}
\newcommand{\mz}[1]{\footnote{{\bf\color{blue} Marco: #1}}}
\newcommand{\db}[1]{\footnote{{\bf\color{green!50!black} Dominik: #1}}} 
\newcommand{\st}[1]{\footnote{{\bf\color{purple!90!black} Sebastian: #1}}}
\newcommand{\rj}[1]{\footnote{{\bf\color{orange!50!black} Romain: #1}}}

}
{
\newcommand{\fm}[1]{}
\newcommand{\mz}[1]{}
\newcommand{\db}[1]{}
\newcommand{\st}[1]{}
\newcommand{\rj}[1]{}

}

\graphicspath{{./img/}}
\usepackage{caption}
\usepackage{subcaption}

\DeclareSIUnit\dBm{dBm}
\newcommand{\meter}{\ensuremath{\,\text{m}}\xspace}
\newcommand{\s}{\ensuremath{\,\text{s}}\xspace}
\newcommand{\ms}{\ensuremath{\,\text{ms}}\xspace}
\newcommand{\us}{\ensuremath{\,\mu\text{s}}\xspace}

\newcommand{\MHz}{\ensuremath{\,\text{MHz}}\xspace}

\newcommand{\dBm}{\ensuremath{\,\text{dBm}}\xspace}
\newcommand{\ppm}{\ensuremath{\,\text{ppm}}\xspace}
\newcommand{\kbps}{\ensuremath{\,\text{kbit/s}}\xspace}

\newcommand{\percent}{\ensuremath{\,\text{\%}}\xspace}

\newcommand{\V}{\ensuremath{~\text{V}}\xspace}

\newcommand{\fakepar}[1]{\vspace{1mm}\noindent\textbf{#1.}}

\newcommand{\cps}{CPS\xspace}

\newcommand{\ttw}{TTW\xspace}

\newcommand{\ap}{AP\xspace}
\newcommand{\cp}{CP\xspace}
\newcommand{\dpp}{DPP\xspace}
\newcommand{\bolt}{Bolt\xspace}
\newcommand{\lti}{LTI\xspace}

\newcommand{\sync}{SYNC\xspace}

\newcommand{\eREF}{\ensuremath{e_{\mathit{ref}}}\xspace}
\newcommand{\eREFmax}{\ensuremath{\hat{e}_{\mathit{ref}}}\xspace}

\newcommand{\eSYNCmax}{\ensuremath{\hat{e}_{\mathit{SYNC}}}\xspace}

\newcommand{\eTASKmax}{\ensuremath{\hat{e}_{\mathit{task}}}\xspace}

\newcommand{\rhoAP}{\ensuremath{\rho_{\mathit{\ap}}}\xspace}
\newcommand{\rhoAPmax}{\ensuremath{\hat{\rho}_{\mathit{\ap}}}\xspace}
\newcommand{\rhoCP}{\ensuremath{\rho_{\mathit{\cp}}}\xspace} 
\newcommand{\rhoCPmax}{\ensuremath{\hat{\rho}_{\mathit{\cp}}}\xspace}

\newcommand{\trefEstim}{\ensuremath{\hat{t}_{\mathit{ref}}}\xspace}

\newcommand{\fAP}{\ensuremath{f_{\mathit{\ap}}}\xspace}

\newcommand{\Tjitter}{\ensuremath{\widetilde{T}_{\mathit{end}}}\xspace}
\newcommand{\Jitter}{\ensuremath{J}\xspace}

\newcommand{\Tupdate}{\ensuremath{T_U}\xspace}
\newcommand{\Tdelay}{\ensuremath{T_D}\xspace}

\newcommand{\apnode}[1]{AP\textsubscript{#1}}
\newcommand{\cpnode}[1]{CP\textsubscript{#1}}

\usepackage{fancyhdr}
\newcommand{\mytitle}{\textbf{Accepted final version.}
To appear in \textit{ACM Transactions on Cyber-Physical Systems}.\\
\copyright 2019 Copyright is with the authors, exclusive licensing agreement with ACM.}
\begin{document}
\title[Fast Feedback Control over Multi-hop Wireless Networks]{Fast Feedback Control over Multi-hop Wireless Networks with Mode Changes and Stability Guarantees}

\author{Dominik Baumann}
\authornote{Equal contribution}
\affiliation{%
  \institution{Max Planck Institute for Intelligent Systems}
  \streetaddress{Max-Planck-Ring 4}
  \postcode{72076}
  \city{T{\"u}bingen}
  \country{Germany}}
\email{dominik.baumann@tuebingen.mpg.de}
\author{Fabian Mager}
\authornotemark[1]
\affiliation{%
  \institution{TU Dresden}
  \streetaddress{Helmholtzstra{\ss}e 18}
  \postcode{01062}
  \city{Dresden}
  \country{Germany}
}
\email{fabian.mager@tu-dresden.de}
\author{Romain Jacob}
\affiliation{%
  \institution{ETH Zurich}
  \streetaddress{Gloriastrasse 35}
  \postcode{8092}
  \city{Zurich}
  \country{Switzerland}
  }
\email{romain.jacob@tik.ee.ethz.ch}
\author{Lothar Thiele}
\affiliation{%
  \institution{ETH Zurich}
  \streetaddress{Gloriastrasse 35}
  \postcode{8092}
  \city{Zurich}
  \country{Switzerland}
}
\email{thiele@ethz.ch}
\author{Marco Zimmerling}
\affiliation{%
  \institution{TU Dresden}
  \streetaddress{Helmholtzstra{\ss}e 18}
  \postcode{01062}
  \city{Dresden}
  \country{Germany}
  }
\email{marco.zimmerling@tu-dresden.de}
\author{Sebastian Trimpe}
\affiliation{%
  \institution{Max Planck Institute for Intelligent Systems}
  \streetaddress{Heisenbergstrasse 3}
  \postcode{70569}
  \city{Stuttgart}
  \country{Germany}}
  \email{trimpe@is.mpg.de}


\begin{abstract}
Closing feedback loops fast and over long distances is key to emerging cyber-physical applications; for example, robot motion control and swarm coordination require update intervals of tens of milliseconds.
Low-power wireless communication technology is preferred for its low cost, small form factor, and flexibility, especially if the devices support multi-hop communication.
Thus far, however, feedback control over multi-hop low-power wireless networks has only been demonstrated for update intervals on the order of seconds.
To fill this gap, this paper presents a wireless embedded system that supports dynamic mode changes and tames imperfections impairing control performance (\eg jitter and message loss), and a control design that exploits the essential properties of this system to provably guarantee closed-loop stability for physical processes with linear time-invariant dynamics in the presence of mode changes.
Using experiments on a cyber-physical testbed with 20 wireless devices and multiple cart-pole systems, we are the first to demonstrate and evaluate feedback control and coordination with mode changes over multi-hop networks for update intervals of 20 to 50 milliseconds.
\end{abstract}



\maketitle

\thispagestyle{fancy}	
\pagestyle{empty}

\renewcommand{\shortauthors}{D. Baumann et al.}


\section{Introduction}
\label{sec:intro}

Cyber-physical systems~(\cps) rely on embedded computers and networks to monitor and control physical systems~\cite{Derler2012}.
While monitoring using \emph{sensors} allows, for example, to better understand environmental processes~\cite{Corke2010}, it is feedback control and coordination of possibly multiple physical systems through \emph{actuators} what nurtures the \cps vision of robotic materials~\cite{Correll2017}, smart transportation~\cite{Besselink2016}, multi-robot swarms for disaster response and manufacturing~\cite{Hayat2016}, \etc

A key hurdle to realizing the \cps vision is how to close the \emph{feedback loops} between sensors and actuators as these may be numerous, mobile, distributed across large physical spaces, and attached to devices subject to size, weight, and cost constraints.
Wireless multi-hop communication among low-power, possibly battery-supported nodes\footnote{While actuators often require wall power, low-power operation is crucial for sensors and controllers, which may run on batteries and harvest energy from various sources, such as solar in outdoor scenarios or machine vibrations in a factory~\cite{Akerberg2011}.} offers the cost efficiency and flexibility to overcome this hurdle~\cite{Lu2016a,Watteyne2016} if two requirements are met.
First, fast feedback is required to keep up with the dynamics of physical systems~\cite{Astrom1996}; for example, robot-motion control and drone-swarm coordination require update intervals of tens of milliseconds~\cite{Preiss2017,Abbenseth2017}.
Second, as feedback control modifies the dynamics of physical systems~\cite{Astrom2008}, guaranteeing \emph{closed-loop stability} in the presence of imperfect wireless communication is essential.
What is more, dynamic changes in the configuration or behavior of the application are a major concern.
For instance, a drone swarm may act differently during take-off, normal flight, and landing (enabling/disabling coordination, switching between different formations, \etc), or machines may be added or removed in an operational production plant. 
The ability to cater for such runtime adaptability in response to an event from the environment or from within the system, known as \emph{mode~changes}~\cite{Chen2018}, is an important requirement.

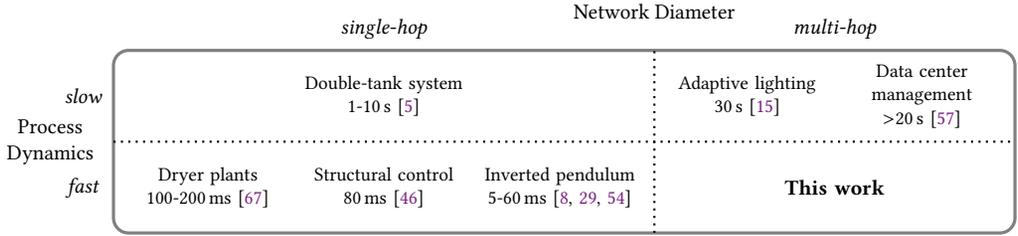
\begin{figure}[!tb]
\centering
\tikzsetnextfilename{cps_solutions}
\tikzexternalexportnextfalse
\begin{tikzpicture}[scale=0.8, every node/.style={scale=0.8}]]
\begin{scope}
\path[clip, preaction={draw, very thick, gray}] [rounded corners=2mm] (-9,-1.5) rectangle (6,1.5);
\draw[thick,dotted] (-9,0) -- (6,0);
\draw[thick,dotted] (0,-1.5) -- (0,1.5);
\end{scope}
\begin{scope}[text=black]
\node[align=center] at (-10.05,0) {Process\\Dynamics};
\node[align=right] at (-9.5,-0.75) {\textit{fast}};
\node[align=right] at (-9.5,0.75) {\textit{slow}};
\node[align=center] at (0,2.15) {Network Diameter};
\node[align=center] at (-4.5,1.85) {\textit{single-hop}};
\node[align=center] at (3,1.85) {\textit{multi-hop}};
\end{scope}
\begin{scope}[text width=9cm,text=black, align=center]
\node[font=\small] (n2) at (-4.5,-0.75) {
	\hspace*{\fill}
	\parbox{2.5cm}{\centering
		Dryer plants\\
		\SIrange[range-units=single, range-phrase=-]{100}{200}{\milli\second} \cite{Ye2001}}
	\hspace*{\fill}
	\parbox{2.5cm}{\centering
		Structural control\\
		\SI{80}{\milli\second} \cite{Lynch2007}}
	\hspace*{\fill}
	\parbox{2.5cm}{\centering
		Inverted pendulum\\
		\SIrange[range-units=single, range-phrase=-]{5}{60}{\milli\second} \cite{Ploplys2004, Hernandez2011, Bauer2014}}
	\hspace*{\fill}
};
\node[font=\small] (n3) at (-4.5,0.75) {
	\hspace*{\fill}
	\parbox{3cm}{\centering
		Double-tank system\\
		\SIrange[range-units=single, range-phrase=-]{1}{10}{\second} \cite{Araujo2014}}
	\hspace*{\fill}
};
\end{scope}
\begin{scope}[text width=6cm,text=black, align=center]
\node[font=\bfseries] (n1) at (3,-0.75) {This work};
\node[font=\small] (n3) at (3,0.75) {
	\hspace*{\fill}
	\parbox{2.5cm}{\centering
		Adaptive lighting\\
		\SI{30}{\second} \cite{Ceriotti2011}}
	\hspace*{\fill}
	\parbox{2.5cm}{\centering
		Data center management\\
		\SI{>20}{\second} \cite{Saifullah2014}}
	\hspace*{\fill}
};
\end{scope}
\end{tikzpicture}
\caption{Design space of wireless \cps that have been validated on real-world devices and wireless networks.}
\label{fig:solutions}
\end{figure}


Hence, this paper investigates the following question: \emph{Is it possible to enable fast feedback control and coordination with mode changes across \textbf{real-world} multi-hop low-power wireless networks while providing formally proven guarantees on closed-loop stability?}
Prior works on control over wireless that validate their design through experiments on physical platforms do not provide an affirmative answer.
As illustrated in \figref{fig:solutions} and detailed in \secref{sec:related}, solutions based on \emph{multi-hop} communication have only been demonstrated for physical systems with \emph{slow} dynamics (\ie update intervals of seconds) and do not provide stability guarantees.
Practical solutions with stability guarantees for \emph{fast} process dynamics (\ie update intervals of tens of milliseconds as typical of, \eg mechanical systems) exist, but these are only applicable to \emph{single-hop} networks and therefore lack the flexibility required by future \cps applications~\cite{Hayat2016,Luvisotto2017}.
None of these solutions considers mode changes and, indeed, there exists no distributed wireless system design with suppport for mode changes to date.

\fakepar{Contribution and road-map}
This paper presents the design, analysis, and real-world validation of a wireless \cps that fills this gap.
\secref{sec:overview} highlights the main challenges and corresponding system design goals we need to achieve when closing feedback loops fast over multi-hop wireless networks while supporting mode changes.
Underlying our approach is a careful co-design of the wireless embedded components (in terms of hardware and software) and the closed-loop control system, as detailed in Sections~\ref{sec:embedded} and \ref{sec:control}.
We tame typical wireless network imperfections, such as message loss and jitter, so that they can be tackled by well-known control techniques or safely neglected.
As a result, our design is amenable to a formal end-to-end analysis of all \cps components (\ie wireless embedded, control, and physical systems), which we exploit to provably guarantee closed-loop stability for physical systems with \emph{linear time-invariant~(\lti)} dynamics in the presence of mode changes.
Further, unlike prior work, our solution supports control and coordination of multiple physical systems out of the box, which is a key asset in many \cps applications~\cite{Hayat2016,Preiss2017,Abbenseth2017}.

To evaluate our design in \secref{sec:eval}, we have developed a cyber-physical testbed consisting of 20 wireless embedded devices forming a three-hop network and multiple cart-pole systems whose dynamics match a range of real-world mechanical systems~\cite{Astrom2008,Trimpe2012}.
As such, this testbed addresses an important need in \cps research~\cite{Lu2016a}.
Our experiments reveal the following key findings:
(\emph{i})~two inverted pendulums can be concurrently and safely stabilized by one or two remote controllers across the three-hop wireless network;
(\emph{ii})~the movement of five cart-poles can be synchronized reliably over the network;
(\emph{iii})~our system can safely change between different synchronization and stabilization tasks at runtime;
(\emph{iv})~increasing message loss, update intervals, and mode-change rates can be tolerated at reduced control performance;
(\emph{v})~the experiments confirm our theoretical~results.


In summary, this paper contributes the following:
\begin{itemize}
 \item We are the first to demonstrate feedback control and coordination over real-world multi-hop low-power wireless networks at update intervals of 20 to 50 milliseconds.
 \item We present the first practical wireless \cps design that supports timely and safe mode changes.
 \item We provide conditions to formally verify end-to-end closed-loop stability of our wireless \cps design for physical systems with \lti dynamics in the presence of noise and mode changes.
 \item Extensive experiments on a novel cyber-physical testbed show that our solution can stabilize and synchronize multiple inverted pendulums despite significant message loss.
\end{itemize}

This article significantly extends~\cite{mager2019} by (\emph{i}) adding support for mode changes in the wireless embedded system design (\secref{sec:design_mode_switches}), (\emph{ii}) incorporating mode changes in the stability analysis (\secref{sec:ctrl_mode_changes}), and (\emph{iii}) reporting on new experiments (Sections~\ref{sec:eval_mode_changes} and \secref{sec:eval_dwell_time}).
Moreover, the stability analysis in \secref{sec:stabAnalysis} has been extended to account for process and measurement noise.

\section{Related Work}
\label{sec:related} 

Feedback control over wireless networks has been extensively studied.
The control community has investigated design and stability analysis for wireless (and wired) networks based on different system architectures, delay models, and message loss processes;
recent surveys provide an overview of this large body of fundamental research~\cite{Hespanha2007,Zhang2013}.
However, the majority of those works focuses on theoretical analyses or validates new wireless \cps designs (\eg based on WirelessHART~\cite{Li2016,Ma2018}) only in simulation, thereby ignoring many fundamental challenges that may complicate or prevent a real implementation~\cite{Lu2016a}.
One of the challenges, as detailed in \secref{sec:overview}, is that even slight variations in the quality of a wireless link can trigger drastic changes in the routing topology~\cite{Ceriotti2011}---and this can happen several times per minute~\cite{Gnawali2009}.
Hence, to establish trust in mission-critical feedback control over wireless, a real-world validation against these \emph{dynamics} on a realistic \cps testbed is absolutely essential~\cite{Lu2016a}, as opposed to considering setups with a \emph{statically configured} routing topology and only a few nodes on a desk (as, \eg in~\cite{Schindler2017}).

\figref{fig:solutions} classifies control-over-wireless solutions that have been validated using experiments on physical platforms and against the dynamics of real wireless networks along two dimensions: the network diameter (\emph{single-hop} or \emph{multi-hop}) and the dynamics of the physical system (\emph{slow} or \emph{fast}).  While not representing absolute categories,
we use slow to refer to update intervals on the order of seconds, which is typically insufficient for feedback control of, for example, mechanical systems.

In the \emph{single-hop/slow} category, Araujo \etal \cite{Araujo2014} investigate resource efficiency of aperiodic control with closed-loop stability in a single-hop wireless network of IEEE 802.15.4 devices.
Using a double-tank system as the physical process, update intervals of one to ten seconds are sufficient.

A number of works in the \emph{single-hop/fast} class stabilize an inverted pendulum via a controller that communicates with a sensor-actuator node at the cart.
The update interval is \SI{60}{\ms} or less, and the interplay of control and network performance, as well as closed-loop stability are investigated for different wireless technologies: Bluetooth~\cite{Eker2001}, IEEE 802.11~\cite{Ploplys2004}, and IEEE 802.15.4~\cite{Bauer2014,Hernandez2011}.
Belonging to the same class, Ye \etal use three IEEE 802.11 nodes to control two dryer plants at update intervals of \SIrange[range-units=single, range-phrase=-]{100}{200}{\ms}~\cite{Ye2001}, and Lynch \etal use four proprietary wireless nodes to demonstrate control of a three-story test structure at an update interval of \SI{80}{\ms}~\cite{Lynch2008}.

As for \emph{multi-hop} networks, there are only solutions for \emph{slow} process dynamics without stability analysis.
For example, Ceriotti \etal study adaptive lighting in road tunnels~\cite{Ceriotti2011}.
The length of the tunnels makes multi-hop communication unavoidable, yet the required update interval of \SI{30}{\second} enables a reliable solution based on existing concepts.
Similarly, Saifullah \etal present a multi-hop solution for power management in data centers, using update intervals of \SI{20}{\second} or greater~\cite{Saifullah2014}.

In contrast to these works, as illustrated in \figref{fig:solutions}, in this paper we demonstrate \emph{fast} feedback control over low-power wireless \emph{multi-hop} networks (IEEE 802.15.4) at update intervals of \SIrange[range-units=single, range-phrase=-]{20}{50}{\ms}, which is
significantly faster than existing multi-hop solutions.
Moreover, we provide a formal stability proof, and our solution seamlessly supports both control and coordination of multiple physical systems, which we validate through experiments on a real-world cyber-physical testbed.

None of these references considers mode changes to adapt to, for example, changing application requirements and operating conditions to efficiently use the limited available resources.
Systems that change between different modes have been studied in the control community under the term \emph{switched systems}~\cite{liberzon1999basic,branicky1998multiple}.
Analyzing the stability of a switched system is difficult as even switching between stable subsystems may lead to an unstable overall system.
There exists also a large body of work on multi-mode systems in the real-time literature, developing different task models~\cite{Buttazzo1998}, analysis techniques~\cite{Phan2008}, and mode-change protocols~\cite{Chen2018}.
However, most of these efforts lack an experimental evaluation, and none of them tackles the challenges of a distributed wireless system.


\section{Problem Formulation and Approach}
\label{sec:overview}


Guaranteeing closed-loop stability for fast feedback control over multi-hop wireless networks in the presence of mode changes is an unsolved problem.
The problem originates from the scenarios targeted by emerging \cps applications.
This section distilles the main characteristics of the target scenarios, discusses the associated challenges, and outlines our approach to address those challenges.

\fakepar{Scenario}
We consider multi-mode wireless \cps consisting of embedded devices (\emph{nodes}) with low-power wireless radios distributed across physical space.
The nodes execute different \emph{application tasks} (\ie sensing, control, or actuation) that can exchange \emph{messages} with each other over a multi-hop wireless network.
Each node can execute multiple application tasks, which may belong to different feedback loops closed over the same wireless network.
%
\begin{figure}[!tb]
	\centering
	\includegraphics[width=0.6\linewidth]{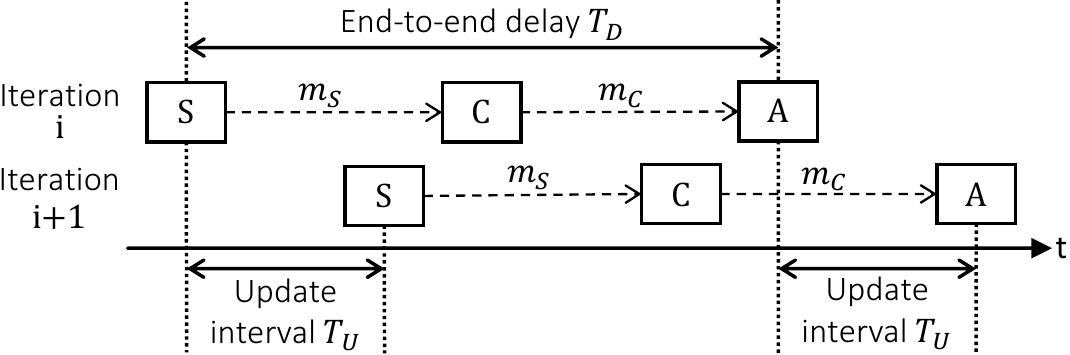}
	\caption{Tasks and messages of a feedback loop with a remote controller.
	\capt{In every iteration, the sensing task takes a measurement of the physical system and sends it to the control task, which computes a control signal and sends it to the actuation task.}}
	\label{fig:overview}
\end{figure}

As an example, \figref{fig:overview} depicts the execution of a given set of application tasks and the exchange of messages between them for a single periodic feedback loop with one sensor and one actuator.
As visible in the figure, the \emph{update interval}~\Tupdate is the time between consecutive sensing or actuation tasks, while the \emph{end-to-end delay}~\Tdelay is the time between corresponding sensing and actuation tasks.

At runtime, the wireless \cps can dynamically switch from one well-defined \emph{mode} to another, either in response to an event from the environment or in response to an event from within the system.
Example events include switching from system start-up to normal operation, the addition/removal of wireless embedded devices or physical systems, and the failure of hardware/software components.
Therefore, a mode change may involve a change in the set of physical systems and the set of devices the wireless \cps is composed of, in the set of application tasks executed by the devices including the control tasks and its parameters such as \Tupdate and \Tdelay, and in the amount of messages exchanged among the application tasks per unit of time.   
In this way, the wireless \cps adapts to changing application requirements and operating conditions to achieve the desired functionality and efficiency.


\fakepar{Challenges}
Fast feedback control over wireless multi-hop networks is an open problem to date due to the following fundamental challenges:
\begin{itemize}
 \item \emph{Lower end-to-end throughput.} Multi-hop networks have a lower end-to-end throughput than single-hop networks: Because of interference, the theoretical multi-hop upper bound on the throughput is half the single-hop upper bound~\cite{Osterlind2008}.
 This limits the amount of sensor readings and control signals that can be exchanged within a given maximum update interval.
 \item \emph{Significant delays and jitter.} Multi-hop networks also incur larger end-to-end communication delays compared with single-hop networks, and the delays are subject to larger variations because of message retransmissions or dynamically changing routing topologies~\cite{Ceriotti2011}, introducing significant jitter. Delays and jitter can both destabilize a feedback system~\cite{Wittenmark1995,Walsh2002}.
 \item \emph{Constrained traffic patterns.} In a single-hop network, each node can communicate with every other node due to the broadcast nature of the wireless medium. This is generally not the case in a multi-hop network. For example, WirelessHART only supports communication to and from a dedicated gateway that connects the wireless multi-hop network to the control system. Feedback control under constrained traffic patterns is more challenging, and may lead to poor control performance or even infeasibility of closed-loop stability~\cite{Yang2005}.
 \item \emph{Correlated message loss.}  Wireless networks are prone to unpredictable message loss, which complicates control design. Further, wireless interference and other disturbances can cause significant correlation among the losses observed over individual wireless links~\cite{Srinivasan2008},  which makes a valid theoretical analysis to provide strong stability guarantees hard, if not impossible.
 \item \emph{Message duplicates and out-of-order message delivery} are typical of many multi-hop wireless communication protocols~\cite{Gnawali2009,Duquennoy2015}, further hindering control design and stability analysis~\cite{Zhang2013}.
\end{itemize}

These challenges also complicate the execution of mode changes.
To transition from one mode to another in a timely and safe manner, \emph{all} nodes in the system must \emph{synchronously} change their mode.
In the absence of a shared clock, however, the system-wide time synchronization and signaling required to do so are hindered by the limited throughput, unpredictable communication delays, and message loss.
For these reasons, mode changes in a real distributed wireless system have not been demonstrated so far, not to mention formal guarantees on closed-loop stability.

\fakepar{Approach}
The co-design approach we adopt to fill this gap can be summarized as follows: \emph{Address the challenges through the wireless embedded system design to the extent possible, and then consider the resulting key properties in the control design.}
More concretely, we pursue three goals with our design of the wireless embedded hardware and software components:
\begin{itemize}
 \item[\textbf{G1}] reduce and bound imperfections impairing closed-loop stability and control performance (\eg make \Tupdate and \Tdelay as short as possible, and bound the worst-case jitter on both quantities);
 \item[\textbf{G2}] support predictable mode changes and arbitrary traffic patterns in multi-hop low-power wireless networks with real dynamics (\eg time-varying wireless links and network topologies);
 \item[\textbf{G3}] operate efficiently in terms of limited resources (\eg energy, wireless bandwidth, computational power), while accommodating the computational requirements of the controller.
\end{itemize}
On the other hand, our control design aims to achieve the following goals:
\begin{itemize}
 \item[\textbf{G4}] consider all essential properties of the wireless embedded system to guarantee closed-loop stability for the entire \cps across mode changes for physical systems with \lti dynamics;
 \item[\textbf{G5}] enable an efficient implementation of the controller on modern low-power embedded devices;
 \item[\textbf{G6}] exploit the support for arbitrary traffic patterns to straightforwardly solve distributed control tasks (\eg when a local controller requires information about other processes and controllers).
\end{itemize}

In the following, \secref{sec:embedded} describes the design of the wireless embedded system, while \secref{sec:control} details the control design and stability analysis.
Sections~\ref{sec:embedded} and~\ref{sec:control} address different aspects of a mode change, which we clarify and define throughout the discussion.\mz{There are, indeed, different aspects (and interpretations) of a mode change being addressed by the wireless embedded and the control design. In the control design, we defer the clarification until Section 5.4, which makes sense and also worked for the reviewers. Similarly, it makes sense to defer the clarification of what a mode change entails for the wireless embedded design until Section 4.3, because there it becomes relevant for the explanation of the scheduling problem, where each schedule essentially represents a different mode.}  


\section{Wireless Embedded System Design}
\label{sec:embedded}

To achieve goals \textbf{G1}--\textbf{G3}, we design a wireless embedded system consisting of four building blocks:
\begin{itemize}
 \item[1)] a \emph{low-power wireless protocol} that provides multi-hop many-to-all communication with minimal, bounded end-to-end delay and accurate network-wide time synchronization;
 \item[2)] a \emph{hardware platform} that enables a predictable and efficient execution of all application tasks and message transfers;
 \item[3)] a \emph{scheduling framework} to schedule for each mode all application tasks and message transfers so that given bounds on \Tupdate and \Tdelay are met at minimum communication energy costs;
 \item[4)] a \emph{mode-change protocol} to transition in a timely and safe manner between different modes at runtime in response to an event from the environment or from within the system.
\end{itemize}
We describe each building block below, followed by an analysis of the resulting properties that matter for the control design.

\subsection{Low-power Wireless Protocol}
\label{sec:wireless_protocol}

To support arbitrary traffic patterns (\textbf{G2}), we require a multi-hop protocol capable of many-to-all communication.
Moreover, the protocol must be highly reliable and the time needed for many-to-all communication must be tightly bounded~(\textbf{G1}).
It has been shown that a solution based on Glossy floods~\cite{Ferrari2011} can meet these requirements with high efficiency (\textbf{G3}) in the face of significant wireless dynamics (\textbf{G2})~\cite{Zimmerling2017}.
Thus, similar to other recent proposals~\cite{Ferrari2012,Istomin2016}, we design a low-power wireless protocol on top of Glossy, but aim at a new design point:
bounded end-to-end delays of at most a few tens of milliseconds \emph{for the many-to-all exchange of multiple messages} in a control cycle.

\begin{figure}[!tb]
	\centering
	\includegraphics[width=0.6\linewidth]{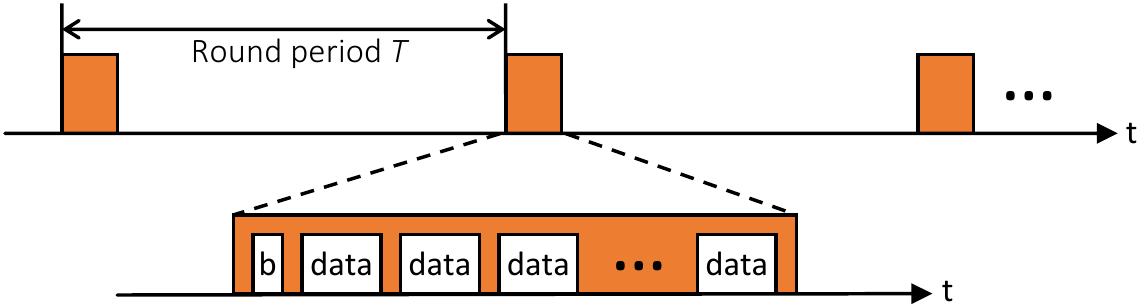}
	\caption{Time-triggered operation of low-power wireless protocol. \capt{Communication occurs in rounds that are scheduled with a given round period $T$. Every beacon (b) and data slot corresponds to a one-to-all Glossy flood~\cite{Ferrari2011}.}} 
	\label{fig:protocol}
\end{figure}

As shown in \figref{fig:protocol}, the operation of the protocol proceeds as a series of \emph{communication rounds} with \emph{period}~$T$.
Each round consists of a sequence of non-overlapping time \emph{slots}.
In every slot, all nodes in the network participate in a Glossy flood, where a message is sent from one node to all others.
Glossy approaches the theoretical minimum latency for one-to-all flooding at a reliability above 99.9\percent, and provides microsecond-level network-wide time synchronization~\cite{Ferrari2011}.
All nodes are synchronized at the beginning of every round during a flood initiated by a dedicated \emph{host} node in the \emph{beacon} slot (b).
Nodes exploit the synchronization to remain in a low-power sleep mode between rounds and to awake in time for the next round, as specified by the round period~$T$.

It is important to note that due to the way Glossy exploits \emph{synchronous transmissions}~\cite{Ferrari2011}, our wireless protocol operates \emph{independently} of the time-varying network topology.
This implies that any logic built atop the wireless protocol, such as a control or scheduling algorithm, need not worry about the state of individual wireless links or the positions of specific nodes in the network.
This is a fundamental difference to wireless protocols based on routing, such as WirelessHART and 6TiSCH. 
As we shall see in the following sections, the network topology independence greatly simplifies control design and allows for providing formally proven guarantees that also hold in practice.


As detailed in \secref{subsec:scheduling}, the communication schedule for each mode is computed offline based on the traffic demands, and is distributed to all nodes before the application operation starts.
A schedule includes the assignment of messages to \emph{data} slots in each round (see \figref{fig:protocol}) and the round period~$T$.
Using static schedules brings several benefits.
First, we can a priori verify if closed-loop stability can be guaranteed for the achievable latencies (see \secref{sec:control}).
Second, compared to prior solutions~\cite{Ferrari2012,Istomin2016,Zimmerling2017,Jacob2016}, we can support significantly shorter latencies, the protocol is more energy efficient (no need to send schedules), and more reliable (schedules cannot be lost over wireless).

\subsection{Hardware Platform}
\label{subsec:embeddedHardware}

\cps devices need to concurrently handle (possibly multiple) application tasks and message transfers.
While message transfers involve little but frequent computations (\eg serving interrupts from the radio hardware), sensing and especially control tasks may require less frequent but more demanding computations (\eg floating-point operations).
An effective approach to achieve low latency and high energy efficiency for such diverse workloads is to exploit hardware heterogeneity~(\textbf{G3}).

For this reason, we leverage a heterogeneous \emph{dual-processor platform (\dpp)}.
Application tasks execute exclusively on a 32-bit MSP432P401R ARM Cortex-M4F \emph{application processor (\ap)} running at 48\MHz, while the wireless protocol executes on a dedicated 16-bit CC430F5147 \emph{communication processor (\cp)} running at 13\MHz.
The \ap has a floating-point unit and a rich instruction set, which facilitate computations related to sensing and control.
The \cp features a low-power microcontroller and a low-power wireless radio operating at 250\kbps in the 868\MHz frequency band.

\ap and \cp are interconnected using \bolt~\cite{Sutton2015a}, an ultra-low-power processor interconnect that supports asynchronous bidi\-rec\-tion\-al message passing with formally verified worst-case execution times.
\bolt decouples the two processors with respect to time, power, and clock domains, enabling energy-efficient concurrent executions with only small and bounded interference, thereby limiting jitter and preserving the time-sensitive operation of the wireless protocol.

All {\cp}s are time-synchronized via the wireless protocol.
In addition, \ap and \cp must be synchronized locally to minimize end-to-end delay and jitter among application tasks running on different {\ap}s~(\textbf{G1}).
To this end, we use a GPIO line between the processors, called \emph{\sync} line.
Every \cp asserts the \sync line in response to an update of Glossy's time synchronization.
Every \ap  resynchronizes its local time base and schedules application tasks and message passing over \bolt with specific offsets relative to those \sync line events.
Likewise, the {\cp}s execute the current communication schedule and perform \sync line assertion and message passing over \bolt with specific offsets relative to the start of communication rounds.
Thus, all {\ap}s and {\cp}s act in concert.

\subsection{Scheduling Framework}
\label{subsec:scheduling}

\begin{figure}[!tb]
	\centering
	\includegraphics[width=1.0\linewidth]{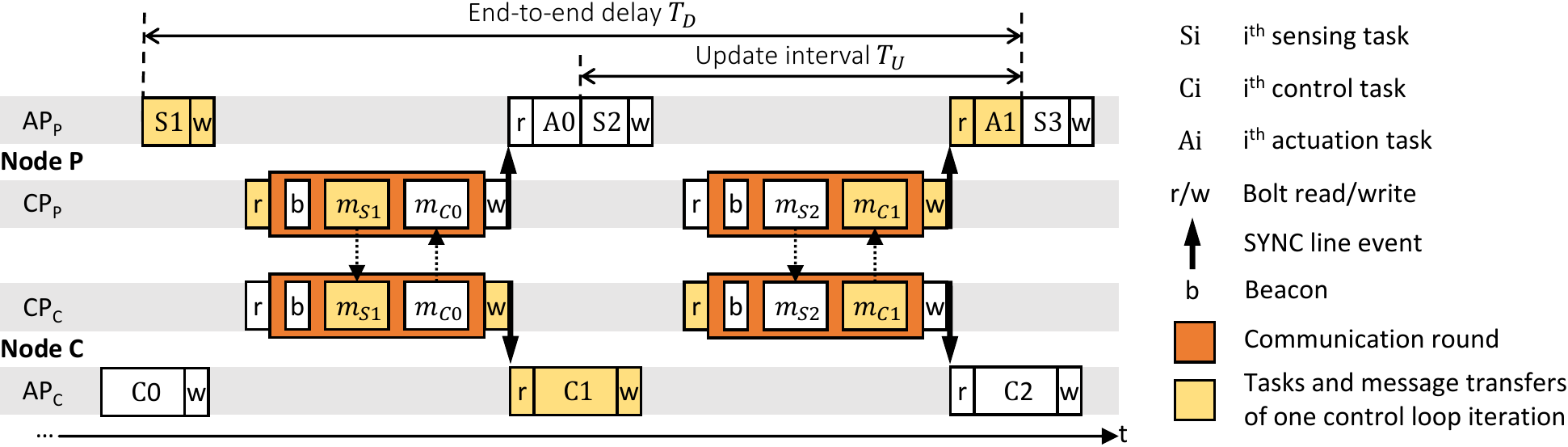}
	\caption{Example schedule of application tasks and message transfers between two {\dpp}s. \capt{Node P senses and acts on a physical system, while node C runs the controller. The update interval \Tupdate is half the end-to-end delay~\Tdelay.}}
	\label{fig:schedule}
\end{figure}

We illustrate the scheduling problem we need to solve with a simple example, where node P senses and acts on a physical system and node C runs the controller.
\figref{fig:schedule} shows a possible schedule of the application tasks and message transfers.
After sensing (S1), \apnode{P} writes a message containing the sensor reading into \bolt (w).
\cpnode{P} reads out the message (r) before the communication round in which that message ($m_{S1}$) is sent using the wireless protocol.
\cpnode{C} receives the message and writes it into \bolt.
After reading out the message from \bolt, \apnode{C} computes the control signal (C1) and writes a message containing it into \bolt.
The message ($m_{C1}$) is sent to \cpnode{P} in the next round, and then \apnode{P} applies the control signal on the physical system (A1).

This schedule resembles a pipelined execution, where in each communication round the last sensor reading and the next control signal (computed based on the previous sensor reading) are exchanged ($m_{S1} \, m_{C0}$, $m_{S2} \, m_{C1}$, $\ldots$).
Note that while it is indeed possible to send the corresponding control signal in the same round ($m_{S1} \, m_{C1}$, $m_{S2} \, m_{C2}$, $\ldots$), doing so would increase the update interval \Tupdate at least by the sum of the execution times of the control task, \bolt read, and \bolt write.
For the example schedule in \figref{fig:schedule}, the update interval \Tupdate is exactly half the end-to-end delay \Tdelay.

%

In general, the scheduling problem entails computing the communication schedule and the offsets with which all {\ap}s and {\cp}s in the system perform wireless communication, execute application tasks, transfer messages over \bolt, and assert the \sync line.
The problem gets extremely complex for realistic scenarios with more devices, physical systems, and feedback loops that are closed over the same wireless network. 
Moreover, whenever there is a change in the set of application tasks or the message transfers between application tasks, this corresponds to a mode change from the perspective of the wireless embedded system, which is associated with a distinct schedule.\mz{I think this sentence clarifies things pretty well and connects nicely with text before and after.}

We leverage Time-Triggered Wireless~(\ttw)~\cite{Jacob2018}, an existing framework tailored to automatically solve this type of scheduling problem.
\ttw assumes that every application task that is active in both the current mode and the next mode is first aborted and then restarted.
For each mode, \ttw takes as main input a dependency graph among the application tasks that are active in this mode and the message transfers between them, similar to \figref{fig:overview}.
Based on an integer linear program, it computes the communication schedules and all offsets mentioned above.
\ttw provides three important guarantees: (\emph{i})~a feasible solution is found if one exists, (\emph{ii})~the solution minimizes the energy consumption for wireless communication, and (\emph{iii})~the solution can additionally optimize user-defined metrics (\eg minimize the update interval \Tupdate \mbox{as for the schedule in~\figref{fig:schedule}).}
Before the wireless \cps is put into operation, the computed schedules are distributed to all devices.
At runtime, the system switches between different schedules using the mode-change protocol described next.

\subsection{Mode-change Protocol}
\label{sec:design_mode_switches}

\begin{figure}[!tb]
	\centering
	\includegraphics[width=0.97\linewidth]{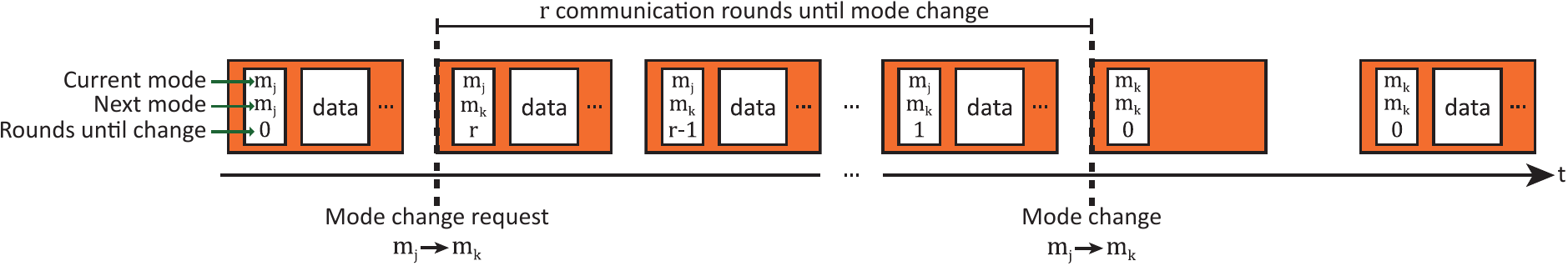}
	\caption{Operation of mode-change protocol. \capt{The beacon slot is used by the host to inform all nodes about the current mode, the next mode, and when to change to the next mode. The first communication round after a mode change contains no data slots because, in general, all messages generated prior to a mode change are meaningless.}}
		\label{fig:modeSwitchSchedule}
\end{figure}

To support predictable transitions between a well-defined set of modes (\textbf{G2}), we design a mode-change protocol whose operation is illustrated in \figref{fig:modeSwitchSchedule}.
Triggered by an event from the environment (\eg operator) or from within the system (\eg detection of a failure), the host sends a \emph{mode-change request} to all nodes by embedding the indices of the current and the next mode, $m_j$ and $m_k$, in the next beacon.
This beacon further contains the number of communication rounds~$r$ until the next mode becomes effective.
In the following $r-1$ rounds, the host keeps sending the mode-change request with a decreasing counter, so that nodes that have not yet received the request or new nodes connecting to the network are also informed of the upcoming mode change.
In turn, the \cp of every node informs the \ap over \bolt.
At the beginning of the communication round in which the counter reaches zero, all nodes that have received the request synchronously change to the new mode.
This involves aborting the old and (re)starting the new application tasks, flushing the \bolt queues, and implementing the new schedule.
Flushing the \bolt queues is appropriate because messages generated in the old mode may be meaningless in the new mode.
For the same reason, the communication round in which the counter reaches zero contains no data slots, as visible in \figref{fig:modeSwitchSchedule}.

Our mode-change protocol ensures timeliness and safety.
It guarantees a deterministic delay from when the first mode-change request with counter $r$ is sent until when the new mode becomes effective.
All nodes that are part of the network when the first mode-change request is sent have $r$ chances to learn about the upcoming mode change.
Hence, every such node changes to the new mode with probability $1 - p^r$, where $p$, the probability of missing a beacon, is typically less than 0.1\percent in practice~\cite{Ferrari2011}.
In the unlikely event that a node does not receive any of the $r$ mode-change requests, our mode-change protocol nevertheless ensures a safe operation.
In particular, a node is only allowed to participate in a communication round if it has received the beacon and the current mode (embedded in the beacon) matches its own local mode.
Therefore, it is guaranteed that a node never disturbes the rest of the system.
If a node misses a mode change, it executes a fall-back mechanism by which it eventually resynchronizes with the network upon the reception of a beacon.
Then, it changes to the new mode involving the steps outlined above, and starts to participate in the following communication round.
The stability guarantee derived in \secref{sec:control} is only applicable to a control loop if all nodes necessary to close that loop are part of the network and in the current mode.\mz{This is the most concise formulation I could come up with.}

While \ttw and our mode-change protocol could be extended, for example, to preserve periodicity of application tasks across mode changes~\cite{Chen2018}, this paper focuses on the interplay of control per\-for\-mance/stability and mode changes when closing feedback loops over multi-hop wireless networks.

\vspace{-0.5mm}
\subsection{Essential Properties and Jitter Analysis}
\label{sec:properties}
\vspace{-0.5mm}

%

The presented wireless embedded system provides the following properties for the control design:
\begin{itemize}
\item[\textbf{P1}] As analyzed below, for update intervals \Tupdate and end-to-end delays \Tdelay up to 100\ms, the worst-case jitter on \Tupdate and \Tdelay is bounded by \SI{\pm50}{\micro\second}.
\Tupdate and \Tdelay are constant for each mode.
\item[\textbf{P2}] Statistical analysis of millions of Glossy floods~\cite{Zimmerling2013} and percolation theory for time-varying networks~\cite{Karschau2018} have shown that the spatio-temporal diversity in a Glossy flood reduces the temporal correlation in the series of received and lost messages by a node, to the extent that the series can be safely approximated by an \iid Bernoulli process.
Using Glossy, the success probability in real multi-hop low-power wireless networks is typically larger than 99.9\percent~\cite{Ferrari2011}.
\item[\textbf{P3}] Because the multi-hop wireless protocol provides many-to-all communication, arbitrary traffic patterns required by the application are efficiently supported.
\item[\textbf{P4}] With high probability all nodes synchronously change mode, and otherwise do not disturb the running system. During a mode change, there is a dead time equal to \Tupdate of the new mode.
\item[\textbf{P5}] It is guaranteed that message duplicates and out-of-order message deliveries do not occur.
\end{itemize}
%

To underpin \textbf{P1}, we analyze the \emph{worst-case} jitter on \Tupdate and \Tdelay.
We refer to \Tjitter as the nominal time interval between the end of two tasks executed on (possibly) different {\ap}s.
Due to jitter \Jitter, this interval may vary, resulting in an actual length of $\Tjitter + \Jitter$.
In our system, the jitter is bounded by
\begingroup
\setlength{\abovedisplayskip}{3pt}
\setlength{\belowdisplayskip}{3pt}
\begin{equation}
	|\, \Jitter \,| \leq 2 \, \left(\eREFmax + \eSYNCmax +  \Tjitter \, (\rhoAPmax + \rhoCPmax) \right) + \eTASKmax,
	\label{eq:jitter_bound}
\end{equation}
\endgroup
where each term on the right-hand side of \eqref{eq:jitter_bound} is detailed below.


\emph{1) Time synchronization error between {\cp}s.}
Using Glossy, each \cp computes an estimate \trefEstim of the reference time~\cite{Ferrari2011} to schedule subsequent activities.
In doing so, each \cp makes an error \eREF with respect to the reference time of the initiator.
Using the approach from~\cite{Ferrari2011}, we measure \eREF for our Glossy implementation and a network diameter of up to nine hops.
Based on 340,000 data points, we find that \eREF ranges always between \SI{-7.1}{\micro\second} and \SI{8.6}{\micro\second}.
We thus consider $\eREFmax=\SI{10}{\micro\second}$ a safe bound for the jitter on the reference time between~{\cp}s.

\emph{2) Independent clocks on \cp and \ap.}
Each \ap schedules activities relative to \sync line events.
As \ap and \cp are sourced by independent clocks, it takes a variable amount of time until an \ap detects that \cp asserted the \sync line.
The resulting jitter is bounded by $\eSYNCmax=1/\fAP$, where $\fAP=48\MHz$ is the frequency of {\ap}s clock.

\emph{3) Different clock drift at {\cp}s and {\ap}s.}
The real offsets and durations of activities on the {\cp}s and {\ap}s depend on the frequency of their clocks.
Various factors contribute to different frequency drifts \rhoCP and \rhoAP, including the manufacturing process, ambient temperature, and aging effects.
State-of-the-art clocks, however, drift by at most $\rhoCPmax = \rhoAPmax = 50\ppm$~\cite{Lenzen2015}.

\emph{4) Varying task execution times.}
The difference between the task's best- and worst-case execution time of the last task in the chain, $\eTASKmax$, adds to the jitter.
For the jitter on the update interval \Tupdate and the end-to-end delay \Tdelay, the last task is the actuation task (see \figref{fig:overview}), which typically exhibits little variance as it is short and highly deterministic.
For example, the actuation task in our experiments has a jitter of $\pm3.4$\us.
To be safe, we consider $\eTASKmax = \SI{10}{\micro\second}$ for our analysis.

Using \eqref{eq:jitter_bound} and the above values, we can compute the worst-case jitter for a given interval \Tjitter.
Fast feedback control as considered in this paper requires $100\ms \geq \Tjitter = \Tdelay > \Tupdate$, which gives a worst-case jitter of \SI{\pm50}{\micro\second} on \Tupdate and \Tdelay, as stated by \textbf{P1}.
\secref{sec:multihop_stabilization} validates this experimentally.

\section{Control Design and Analysis}
\label{sec:control}
Building on the design of the wireless embedded system and its properties \textbf{P1}--\textbf{P5}, this section addresses the design of the control system to accomplish goals \textbf{G4}--\textbf{G6} from \secref{sec:overview}.
Because the wireless system supports arbitrary traffic patterns (\textbf{P3}), various control tasks can be solved including typical single-loop tasks such as stabilization, disturbance rejection, or set-point tracking and multi-agent scenarios such as synchronization, consensus, or formation control.

Here, we focus on remote stabilization over wireless and synchronization of multiple agents as prototypical examples for both the single- and multi-agent case.
For stabilization, the modeling and control design are presented in Sections~\ref{sec:model} and \ref{sec:ctrlDesign}, thus achieving \textbf{G5}. 
The stability analysis for the remote stabilization scenario is provided in \secref{sec:stabAnalysis}. 
The wireless embedded system offers the flexibility to dynamically change between different modes.
In \secref{sec:ctrl_mode_changes}, we extend the stability analysis to also account for mode changes to fulfill \textbf{G4}.
Multi-agent synchronization is discussed in \secref{sec:sync}, highlighting support for straightforward distributed control to achieve \textbf{G6}.

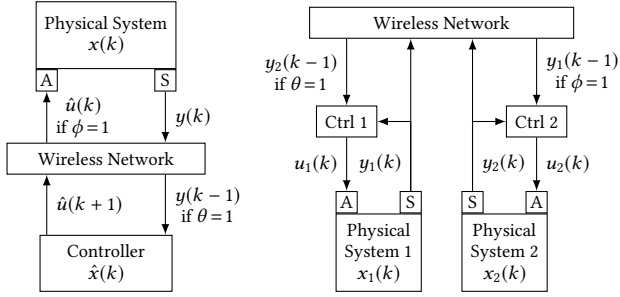
\begin{figure}[!tb]
\centering
\tikzset{radiation/.style={{decorate,decoration={expanding
waves,angle=90,segment length=4pt}}}}
\tikzset{>=latex}
\tikzset{font=\scriptsize}
\begin{tikzpicture}
\node[draw,rectangle, minimum height = 2em,align=center,minimum width = 5em](controller){Controller\\ $\hat{x}(k)$};
.\node[draw, rectangle, minimum width = 7.5em,above = 2.25em of controller, minimum height = 1em](network){Wireless Network};
\node[draw,rectangle,minimum height = 2.5em,above = 2.75em of network,align = center](system){Physical System\\$x(k)$};
\node[draw,rectangle,below=0em of system.south west, anchor = north west,inner sep=0,minimum width=0.8em, minimum height=0.8em](act){A};
\node[draw,rectangle,below=0em of system.south east, anchor = north east,inner sep=0,minimum width=0.8em, minimum height=0.8em](sens){S};
\draw[->] (controller.north-|act)-- node[midway,right]{$\hat{u}(k+1)$} (network.south-|act);
\draw[->] (network.south-|sens)-- node[midway,right, align = center]{$y(k-1)$\\ if $\theta\!=\!1$} (controller.north-|sens);
\draw[->] (network.north-|act)-- node[midway,right,align=center]{$\hat{u}(k)$\\if $\phi\!=\!1$} (act.south);
\draw[->] (sens.south)-- node[midway,right]{$y(k)$} (network.north-|sens);
\end{tikzpicture}
\begin{tikzpicture}
\node[draw,rectangle,minimum height = 2.5em,align = center](system1){Physical\\ System 1\\$x_1(k)$};
\node[draw,rectangle,below=0em of system1.north west, anchor = south west,inner sep=0,minimum width=0.8em, minimum height=0.8em](act1){A};
\node[draw,rectangle,below=0em of system1.north east, anchor = south east,inner sep=0,minimum width=0.8em, minimum height=0.8em](sens1){S};
\node[draw,rectangle, minimum height = 2em,align=center, right = 1.5em of system1](system2){Physical\\ System 2\\ $x_2(k)$};
\node[draw,rectangle,below=0em of system2.north west, anchor = south west,inner sep=0,minimum width=0.8em, minimum height=0.8em](sens2){S};
\node[draw,rectangle,below=0em of system2.north east, anchor = south east,inner sep=0,minimum width=0.8em, minimum height=0.8em](act2){A};
\node[draw, rectangle, minimum width = 10em,above = 8em of $(system1)!0.5!(system2)$, minimum height = 1em](network){Wireless Network};
\node[draw, rectangle, above = 2em of act1](ctrl1){Ctrl 1};
\node[draw,rectangle,above=2em of act2](ctrl2){Ctrl 2};
\draw[->] (sens2)-- node[pos=0.18,right]{$y_2(k)$} (network.south-|sens2);
\draw[->](sens2) |- (ctrl2);
\draw[->] (network.south-|act2)-- node[midway,right,align=center]{$ y_1(k-1)$\\if $\phi\!=\!1$} (ctrl2);
\draw[->](ctrl1) -- node[midway, left]{$u_1(k)$} (act1);
\draw[->] (network.south-|act1)-- node[midway,left,align=center]{$y_2(k-1)$\\if $\theta\!=\!1$} (ctrl1);
\draw[->](ctrl2) -- node[midway,right]{$u_2(k)$}(act2);
\draw[->] (sens1)-- node[pos=0.18,left]{$y_1(k)$} (network.south-|sens1);
\draw[->](sens1) |- (ctrl1);
\end{tikzpicture}
%
\caption{
Considered wireless control tasks: stabilization (left) and synchronization (right).
\capt{LEFT: The feedback loop for stabilizing the physical system is closed over the (multi-hop) low-power wireless network, which induces delay and message losses (captured by \iid Bernoulli variables $\theta$ and $\phi$).  Sensor/actuator and controller are spatially distributed and associated with different nodes.  RIGHT: Two physical systems, each with a local controller (Ctrl), are synchronized over the wireless network.}
}
\label{fig:WirelessControlModel}
\end{figure}

\subsection{Model of Wireless Control System}
\label{sec:model}
We address the remote stabilization task depicted in \figref{fig:WirelessControlModel} (left), where controller and physical system are associated with different nodes, which can communicate via the multi-hop wireless network.
Such a scenario is relevant, for instance, in process control, where the controller often resides at a remote location~\cite{Ma2018}.
We consider stochastic LTI dynamics for the physical process
\begin{subequations}
\label{eqn:sys_complete}
\begin{align}
\label{eqn:ssmodel}
x(k+1) = Ax(k) + Bu(k) + v(k).
\end{align}
The model describes the evolution of the system state $x(k)\in\mathbb{R}^n$ with discrete time index $k \in \mathbb{N}$ in response to control input $u(k)\in\mathbb{R}^m$ and random process noise $v(k)\in\mathbb{R}^n$.
The process noise is (as typical in the literature~\cite{Astrom2008,Hespanha2007}) modeled as an \iid Gaussian random variable with zero mean and variance $\Sigma_\mathrm{proc}$ (\ie $v(k)\sim \mathcal{N}(0,\Sigma_\mathrm{proc})$), capturing, for instance, uncertainty in the model.

We assume that the full system state $x(k)$ can be measured through appropriate sensors, that is,
\begin{align}
\label{eqn:meas}
y(k) = x(k)+w(k),
\end{align} 
\end{subequations}
with sensor measurements $y(k)\in\mathbb{R}^n$ and sensor noise $w(k)\in\mathbb{R}^n$, $w(k)\sim \mathcal{N}(0,\Sigma_\mathrm{meas})$.
If the full state cannot be measured directly, it can often be reconstructed via state estimation techniques \cite{Astrom2008}.

The process model \eqref{eqn:sys_complete} is stated in discrete time. 
This representation is particularly suitable as the wireless system has in each mode a constant update interval \Tupdate with worst case jitter \SI{\pm50}{\micro\second}~(\textbf{P1}), which can be neglected from control's perspective in the considered applications~\cite[p.~48]{Cervin2003}.
Thus, $u(k)$ and $y(k)$ in \eqref{eqn:sys_complete} represent sensing and actuation at periodic intervals \Tupdate as in \figref{fig:schedule}.
As in the example schedule in \figref{fig:schedule}, we consider in the following the relation $\Tdelay \!=\! 2\Tupdate$ between end-to-end delay and update interval.
Nevertheless, we note that our system model, control design, and stability analysis can be readily extended to account for other combinations, including $\Tdelay = n\Tupdate$ ($n\in\mathbb{N}$) as well as non-identical sensor-to-controller and controller-to-actuator delays.

With this, as shown in \figref{fig:WirelessControlModel}, measurements $y(k)$ and control inputs $\hat{u}(k)$ are sent over the wireless network, and both arrive at the controller, respectively system, with a delay of \Tupdate and with a probability governed by two independent Bernoulli processes.
The Bernoulli assumption is very common as it can significantly simplify control design and stability analysis~\cite{Hespanha2007,Zhang2013}, but unlike traditional wireless systems~\cite{Srinivasan2008} this assumption is indeed valid for our wireless embedded design (\textbf{P2}).
We represent the Bernoulli processes by $\theta(k)$ and $\phi(k)$, which are \iid binary variables, indicating lost ($\theta(k) = 0$, $\phi(k)=0$) or successfully received ($\theta(k)=1$, $\phi(k)=1$) messages.
To ease notation and since both variables are \iid, we can omit the time index in the following without any confusion.
We denote the probability of successful delivery by $\mu_\theta$ (\ie $\mathbb{P}[\theta=1] = \mu_\theta$), respectively $\mu_\phi$.
As both, measurements and control inputs, are delayed, it also follows that in case of no message loss, the applied control input $u(k)$ depends on the measurement two steps ago $y(k-2)$.
If a control input message is lost, the input stays constant since zero-order hold is used at the actuator, that is,
\begin{align}
\label{eqn:ZOH}
u(k)=\phi \hat{u}(k)+\left(1-\phi\right)u(k-1).
\end{align}

The model proposed in this section thus captures properties \textbf{P1}, \textbf{P2}, and \textbf{P5}.
While \textbf{P1} and \textbf{P2} are incorporated in the presented dynamics and message loss models, \textbf{P5} means that there is no need to take duplicated or out-of-order sensor measurements and control inputs into account.
Overall, these properties allow for accurately describing the wireless \cps by a fairly straightforward model, which greatly facilitates subsequent control design and analysis.
Property \textbf{P3} is not considered here when dealing with a single control loop, but becomes essential in \secref{sec:sync}.

\subsection{Controller Design}
\label{sec:ctrlDesign}
Designing a feedback controller for the system \eqref{eqn:sys_complete}, we proceed by first discussing state-feedback control for the nominal system (\ie without delays, message loss, and noise), and then enhance the design to cope with the network and sensing imperfections.

\fakepar{Nominal design}
Assuming ideal measurements, we have $y(k)\! =\! x(k)$.
A common strategy in this setting is static state-feedback control, $u(k)\! =\! Fx(k)$, where $F$ is a constant feedback matrix, which can be designed, for instance, via \emph{pole placement} or optimal control such as the \emph{linear quadratic regulator (LQR)} \cite{Astrom2008,Anderson2007}.
Assuming that the system is controllable~\cite{Astrom2008}, which is standard in control design, stable closed-loop dynamics~\eqref{eqn:ssmodel} can be obtained in this way.

\fakepar{Actual design}
We augment the nominal state-feedback design with state predictions to cope with non-idealities, in particular delayed measurements and message loss, as shown in \figref{fig:WirelessControlModel} (left).

Because the measurement arriving at the controller $y(k-1)$ represents information that is one time step in the past, the controller propagates the system for one step as follows:
\begin{align}
\hat{x}(k) &= \theta Ay(k\!-\!1)+(1\!-\!\theta)(A\hat{x}(k\!-\!1))+B\hat{u}(k\!-\!1) \label{eqn:prediction} \\
&= \theta Ax(k\!-\!1) + (1\!-\!\theta)A\hat{x}(k\!-\!1)+B\hat{u}(k\!-\!1)+\theta Aw(k\!-\!1), \nonumber
\end{align}
where $\hat{x}(k)$ is the predicted state, and $\hat{u}(k)$ is the control input
computed by the controller
(to be made precise below).  Both variables are computed by the controller and represent its internal states.  The rationale of \eqref{eqn:prediction} is as follows:  If the measurement message is delivered (the controller has information about $\theta$ because it knows when to expect a message), we compute the state prediction based on this measurement $y(k\!-\!1)\! = \!x(k\!-\!1)+w(k\!-\!1)$; if the message is lost, we propagate the previous prediction $\hat{x}(k\!-\!1)$.
As there is no feedback on lost control messages (\ie about $\phi$) and thus a potential mismatch between the computed input $\hat{u}(k\!-\!1)$ and the actual $u(k\!-\!1)$, the controller can only use $\hat{u}(k\!-\!1)$
in the prediction.

Using $\hat{x}(k)$, the controller has an estimate of the current state of the system.
However, it takes another time step for the currently computed control input to arrive at the physical system.  For computing the next control input, we thus propagate the system another step,
\begin{align}
\label{eqn:pred_inp}
\hat{u}(k+1)&=F\left(A\hat{x}(k)+B\hat{u}(k)\right),
\end{align}
where $F$ is as in the nominal design.
The input $\hat{u}(k+1)$ is sent over the wireless network (see \figref{fig:WirelessControlModel}).

The overall controller design requires only a few matrix multiplications per execution.
This can be efficiently implemented on constrained embedded devices, thus satisfying goal \textbf{G5}.
Moreover, it allows for a formal end-to-end stability analysis as described below, \mbox{thereby satisfying goal \textbf{G4}.}

\subsection{Stability Analysis: Remote Stabilization}
\label{sec:stabAnalysis}
We now prove stability for the closed-loop system given by the dynamic system of \secref{sec:model} and the controller proposed in \secref{sec:ctrlDesign}.  
The model in \secref{sec:model} accounts for the physical process and the essential properties of the wireless embedded system.  The stability proof in this section thus guarantees stability for the remote stabilization scenario.
In \secref{sec:ctrl_mode_changes}, we show that this also guarantees stability under mode changes if the system stays in each mode for sufficient time.

While the process dynamics \eqref{eqn:sys_complete} are time invariant, the message loss introduces time variation and randomness into the system dynamics.  
Therefore, we will leverage stability results for linear, stochastic, time-varying systems~\cite{Boyd1994}.
We first introduce a few required definitions and preliminary results, and then apply these to our problem.
Consider the system
\begin{align}
\label{eqn:gen_sys_noise}
z(k+1)=\tilde{A}(k)z(k) + \tilde{E}(k)\epsilon(k),
\end{align}
with state $z(k)\in\mathbb{R}^n$, $\tilde{A}(k) = \tilde{A}_0 + \sum_{i=1}^L \tilde{A}_ip_i(k)$ and $\tilde{E}(k)=\tilde{E}_0+\sum_{i=1}^L\tilde{E}_ip_i(k)$,
where  $p_i(k)$ are \iid random variables with mean $\E[p_i(k)]=0$, variance $\Var[p_i(k)]=\sigma_{p_i}^2$, and $\E[p_i(k)p_j(k)]=0 \,\forall i,j$; and $\epsilon(k)$ is a vector of \iid Gaussian random variables representing process and measurement noise with $\E[\epsilon(k)]=0$, $\E[\epsilon(k) \epsilon^\mathrm{T}(k)]=W$, and $\epsilon(k)$ independent over time and from $p_i(k)$ and $z(k)$ at every $k$. 
%
A common stability notion for stochastic systems like~\eqref{eqn:gen_sys_noise} is mean-square stability:
\begin{defi}[\cite{freudenberg2010stabilization, zaidi2014stabilization}]
\label{def:MSS}
Let $M(k) := \E[z(k)z^\mathrm{T}(k)]$ denote the state correlation matrix. The system~\eqref{eqn:gen_sys_noise} is \emph{mean square stable (MSS)} if $\, \lim_{k\to\infty}M(k)<\infty$ for any initial $z(0)$.
\end{defi}
For a system that is constantly perturbed by Gaussian noise $\epsilon(k)$, the state correlation does not vanish, but mean-square stability guarantees that it is bounded.
If, however, $\epsilon(k)=0\,\forall k$, we can guarantee the state correlation matrix to vanish.
For systems without noise, we thus employ the following, stricter definition of mean-square stability:
\begin{defi}[{\cite[p.~131]{Boyd1994}}]
\label{def:MSS_noise_free}
The system \eqref{eqn:gen_sys_noise} with $\epsilon(k)=0\,\forall k$ is \emph{MSS} if $\,\lim_{k\to\infty} M(k) = 0$ for any initial $z(0)$.
\end{defi}
Mean-square stability then means $z(k) \to 0$ almost surely as $k\to\infty$, \cite[p.~131]{Boyd1994}.
For ease of exposition, we start by considering \eqref{eqn:gen_sys_noise} without noise, and then extend the analysis to account for noise. 


In control theory, linear matrix inequalities (LMIs) are often used as computational tools to check for system properties such as stability (see \cite{Boyd1994} for an introduction and details).  
For mean-square stability, we employ the following LMI stability result:
\begin{lem}[{\cite[p.~131]{Boyd1994}}]
\label{lem:LMIcond}
System~\eqref{eqn:gen_sys_noise} with $\epsilon(k)=0\,\forall k$ is MSS in the sense of Definition~\ref{def:MSS_noise_free} if, and only if, there exists a positive definite matrix
$P>0$ such that
\begin{equation}
\label{eqn:LMI}
\tilde{A}_0^\mathrm{T}P\tilde{A}_0-P+\sum\nolimits_{i=1}^L\sigma_{p_i}^2\tilde{A}^\mathrm{T}_iP\tilde{A}_i <0 .
\end{equation}
\end{lem}

We now apply this result to the system and controller from \secref{sec:model} and \secref{sec:ctrlDesign}.
The closed-loop dynamics are given by \eqref{eqn:sys_complete}--\eqref{eqn:pred_inp}, which we rewrite
as an augmented system
\begin{align}
\label{eqn:matrix_repr}
\underbrace{\begin{pmatrix}
x(k+1)\\\hat{x}(k+1)\\u(k+1)\\ \hat{u}(k+1)
\end{pmatrix}}_{z(k+1)}
&=\underbrace{\begin{pmatrix}
A&0&B&0\\
\theta A&(1-\theta)A&0&B\\
0&\phi FA&(1-\phi)I&\phi FB\\
0&FA&0&FB
\end{pmatrix}}_{\tilde{A}(k)}
\underbrace{\begin{pmatrix}
x(k)\\ \hat{x}(k)\\u(k)\\ \hat{u}(k)
\end{pmatrix}}_{z(k)}.
\end{align}
The system has the form of \eqref{eqn:gen_sys_noise} with $\epsilon(k)=0$; the transition matrix depends on $\theta$ and $\phi$, and thus on time (omitted for simplicity).  
We can thus apply Lemma~\ref{lem:LMIcond} to obtain the stability result.
\begin{theo}
\label{thm:MSSourSystem}
The system~\eqref{eqn:matrix_repr} is MSS in the sense of Definition~\ref{def:MSS_noise_free} if, and only if, there exists a $P>0$ such that \eqref{eqn:LMI} holds with $L=2$ and
\begin{align*}
\tilde{A}_0 &= \begin{pmatrix}
A&0&B&0\\
\mu_\theta A&(1-\mu_\theta)A&0&B\\
0&\mu_\phi FA&(1-\mu_\phi)I&\mu_\phi FB\\
0&FA&0&FB
\end{pmatrix}, \quad
\tilde{A}_1 = \begin{pmatrix}
0&0&0&0\\
-\mu_\theta A&\mu_\theta A&0&0\\
0&0&0&0\\
0&0&0&0
\end{pmatrix}, \\
\tilde{A}_2 &= \begin{pmatrix}
0&0&0&0\\
0&0&0&0\\
0&-\mu_\phi FA&\mu_\phi I & -\mu_\phi FB\\
0&0&0&0
\end{pmatrix}, \quad
\sigma_{p_1}^2 = \sfrac{1}{\mu_\theta}-1, \quad
\sigma_{p_2}^2 = \sfrac{1}{\mu_\phi} -1.
\end{align*}
\end{theo}
\begin{proof}
For clarity, we reintroduce time indices for $\theta$ and $\phi$ in this proof.
Following a similar approach as in~\cite{Rich2015}, we transform $\theta(k)$ as $\theta(k) = \mu_\theta\left(1-\delta_\theta(k)\right) $ with the new binary random variable
$\delta_\theta(k)\in\{1,1-\sfrac{1}{\mu_\theta}\}$ with
$\mathbb{P}[\delta_\theta(k)=1] = 1-\mu_\theta$ and $\mathbb{P}[\delta_\theta(k)=1-\sfrac{1}{\mu_\theta}] = \mu_\theta$; and analogously for $\phi(k)$ and $\delta_\phi(k)$.
We thus have that $\delta_\theta(k)$ is \iid (because $\theta$ is \iid) with $\E[\delta_\theta(k)] = 0$ and $\Var[\delta_\theta(k)] = \sigma_{p_1}^2$, and similarly for $\delta_\phi(k)$.  
Employing this transformation, $\tilde{A}(k)$ in \eqref{eqn:matrix_repr} is rewritten as $\tilde{A}(k) = \tilde{A}_0 + \sum_{i=1}^2 \tilde{A}_i p_i(k)$ with $p_1(k) = \delta_\theta(k)$, $p_2(k)=\delta_\phi(k)$, and $\tilde{A}_i$ as stated above.  
Thus,  all properties of \eqref{eqn:gen_sys_noise} are satisfied, and Lemma \ref{lem:LMIcond} yields the result.
\end{proof}

Using Theorem~\ref{thm:MSSourSystem}, we can analyze stability for any concrete physical system
\eqref{eqn:sys_complete} (noise-free), a state-feedback controller $F$, and probabilities $\mu_\theta$ and $\mu_\phi$.
Note that the matrices $A$, $B$, and $F$ depend on the length of a discrete time step $k\to k+1$ and thus account for the temporal dynamics of the physical system and the control loop (\ie update interval \Tupdate).
Searching for a $P>0$ that satisfies the LMI~\eqref{eqn:LMI} can be done using efficient numerical tools based on convex optimization (\eg \cite{Labit2002}).  

As it turns out, if such a $P$ is found, 
this also implies stability
for the system defined in~\eqref{eqn:sys_complete}--\eqref{eqn:pred_inp} (with noise), as we state in the following theorem (proof is given in \secref{sec:noise}):
\begin{theo}
\label{thm:MSSourSysNoise}
The system defined in~\eqref{eqn:sys_complete}--\eqref{eqn:pred_inp} is MSS in the sense of Definition~\ref{def:MSS} if the conditions of Theorem~\ref{thm:MSSourSystem} are fulfilled.
\end{theo}


\subsection{Stability Analysis: Remote Stabilization with Mode Changes}
\label{sec:ctrl_mode_changes}
To account for the required adaptability in \cps applications, the wireless embedded system includes support for dynamically changing between different modes at runtime, as stated in \textbf{P4}.
We now extend the stability analysis from the last section to guarantee stability in the presence of mode changes.
From the perspective of control, mode changes correspond to the dynamic system in \eqref{eqn:gen_sys_noise} switching between different modes, which thus becomes 
a stochastic \emph{switched linear system}. 

It is well known~\cite{lin2009stability} that even if each subsystem is stable individually, a switched linear system can, in general, become unstable 
under arbitrary switching.
This can be seen as follows:
An asymptotically stable linear system approaches its equilibrium at $x=0$ for any initial condition $x(0)$.
This, however, does not necessarily happen monotonically.
During the transient behavior in the beginning, the system state may also (and often does) grow.
Thus, switching repeatedly at the wrong moments might lead to the system state growing without bounds. 
Therefore, we have to enhance the stability analysis from \secref{sec:stabAnalysis} to prove stability also under mode changes.


To account for mode changes, we extend the system description in~\eqref{eqn:gen_sys_noise} as follows
\begin{align}
\label{eqn:switched_sys}
z(k+1) = \tilde{A}_{\sigma(k)}(k)\, z(k) + \tilde{E}\,\epsilon(k),
\end{align}
where $\sigma(k)$ is the \emph{switching signal} taking values from the finite set $\mathcal{F}=\{1,\ldots,N\}$, with $N>1$ the number of modes.
As described in \secref{sec:overview}, different modes may correspond to, for instance, different number and dynamics of physical systems, different controllers, \etc
Because of the way $\tilde{A}$ is constructed, such changes are captured by the different matrices $\tilde{A}_{\sigma(k)}(k)$ in \eqref{eqn:switched_sys}.
If the delay requirements or the amount of messages that need to be communicated per round change from one mode to the next, the update interval \Tupdate changes as well.
This means that for different modes the length of a discrete time-step $k\to k+1$ can be different.
When staying in a mode, however, the length of a discrete time-step remains constant (\textbf{P1}).
As the results stated in Lemma~\ref{lem:LMIcond} do not rely on a constant time-step, we can still use them for the further analysis.

In the analysis, we assume that switching is instantaneous.
As per property \textbf{P4}, however, there is always a short dead time when switching to a new mode.
We neglect this for the theoretical analysis and experimentally investigate its effect on control performance and stability in \secref{sec:eval_dwell_time}.

Only in special cases (\eg if the matrices $\tilde{A}_{\sigma(k)}(k)$ commute) stability under arbitrary switching signals $\sigma(k)$ can be guaranteed (see, \eg~\cite{zhai2002qualitative}).
For general systems, different conditions have been established for stability under switching. 
For example, if a system stays in each mode ``long enough,'' stability can be proven \cite{morse1996supervisory,hespanha1999stability}.  
The time a system stays in a mode is called \emph{dwell time}.
Stability can be guaranteed if the dwell time is larger than or equal to some threshold.
In fact, it is sufficient if the system respects this threshold on average.
This is captured by the notion of \emph{average dwell time}:
\begin{defi}[\cite{zhang2008exponential}]
\label{def:avg_dwell_time}
For each switching signal $\sigma(k)$ and any $k_\mathrm{e}>k_\mathrm{s}>k_0$, let $N_{\sigma(k)}(k_\mathrm{s},k_\mathrm{e})$ be the number of switches of $\sigma(k)$ over the interval $[k_\mathrm{s},k_\mathrm{e}]$.
If for any given $N_0>0$, $\tau_\mathrm{a}>0$, we have $N_{\sigma(k)}(k_\mathrm{s},k_\mathrm{e})\leq N_0+(k_\mathrm{e}-k_\mathrm{s})/\tau_\mathrm{a}$, then $\tau_\mathrm{a}$ and $N_0$ are called \emph{average dwell time} and \emph{chatter bound}. 
\end{defi}

We can then give a lower bound on the average dwell time and have the stability guarantee if the switching signal respects this lower bound.
\begin{theo}
\label{thm:switch_system_high_level}
There exists a minimum average dwell time $\tau_\mathrm{a}^*$.
The switched system defined in~\eqref{eqn:switched_sys} is MSS in the sense of Definition~\ref{def:MSS} if the conditions of Theorem~\ref{thm:MSSourSysNoise} hold and the switching signal $\sigma(k)$ is constructed such that $\tau_\mathrm{a}\ge \tau_\mathrm{a}^*$.
\end{theo}

The proof of Theorem~\ref{thm:switch_system_high_level} in the supplementary material derives a bound on the worst-case growth of the system state when switching between modes and a minimum decay rate once the system stays in a mode.
By staying in a mode for long enough on average, we ensure to account for the potential growth during switching and thus guarantee the state of the system to approach its equilibrium. 
In the considered test scenario in \secref{sec:eval_mode_changes}, the dwell times are found to be not overly restrictive.
Moreover, the dwell times only need to be respected \emph{on average}.
The benefit of this can for instance be seen in a scenario, where switches can happen due to external events and be commanded manually.
If during a certain period of time switches faster than the dwell time are necessary because of external events, this can be accounted for by reducing the frequency of manual switches.
Then it may still be possible to respect the dwell time on average.

With Theorem~\ref{thm:switch_system_high_level}, we have the stability guarantee for all cases captured by properties \textbf{P1}--\textbf{P5} of the wireless embedded system and thus achieve goal \textbf{G4}.

\subsection{Multi-agent Synchronization}
\label{sec:sync}

In distributed or decentralized control architectures, different controllers have access to different measurements and inputs,
and thus, in general, different information.  This
is the core reason for why such architectures are more challenging than centralized ones \cite{Lunze1992,Grover2014}.  
For instance, an agent may only be able to communicate point-to-point with its nearest neighbors, or with other agents in a certain range.
Property \textbf{P3} of the wireless embedded system offers a key advantage compared to these structures because every agent in the network has access to all information (except for rare message loss).
We can thus carry out a centralized design, but implement the resulting controllers in a distributed fashion.
Such schemes have been used before for wired-bus networks (\eg in \cite{Trimpe2012}).  Here, we present synchronization of multiple physical systems as an \emph{example} of how distributed control tasks can easily be achieved with the proposed wireless control system, thus achieving \textbf{G6}.

The problem we consider is shown in \figref{fig:WirelessControlModel} (right).
We assume multiple physical processes as in~\eqref{eqn:sys_complete}, but with possibly different dynamics parameters ($A_i$, $B_i$, \etc).
We understand synchronization in this setting as the goal of having the system state of different agents evolve together as close as possible.
That is, we want to keep the error $x_i(k)-x_j(k)$ between the states of systems $i$ and $j$ small.
Instead of synchronizing the whole state vector, also a subset of all states can be considered.
Synchronization of multi-agent systems is a common problem and also occurs under the terms
consensus or coordination~\cite{Lunze2012}.
For simplicity, we consider synchronization of two agents in the following, but the approach directly extends to more than two, as we demonstrate in \secref{sec:eval}.

We consider the architecture in \figref{fig:WirelessControlModel}, where each physical system is associated with a local controller that receives local observations directly and observations from other agents over the multi-hop wireless network.  We present an approach based on an optimal LQR \cite{Anderson2007} to design the synchronizing controllers.
We choose the quadratic cost function
\begin{equation}
\label{eqn:cost}
J = \lim_{K\to\infty}\frac{1}{K}\E\!\Big[\sum\limits_{k=0}^{K-1} \sum_{i=1}^2 \Big(x_i^\mathrm{T}\!(k)Q_i x_i(k) + u_i^\mathrm{T}\!(k)R_i u_i(k) \Big) + (x_1(k)-x_2(k))^\mathrm{T}Q_\text{sync} (x_1(k)-x_2(k)) \Big],
\end{equation}
which expresses our objective of keeping $x_1(k)-x_2(k)$ small (through the weight $Q_\text{sync}>0$) next to usual penalties on states ($Q_i>0$) and control inputs ($R_i>0$).
Using augmented state $\tilde{x}(k) = (x_1(k), x_2(k))^\mathrm{T}$ and input $\tilde{u}(k) = (u_1(k),$ $u_2(k))^\mathrm{T}$, the term in the summation over $k$ becomes 
\begin{align*}
\tilde{x}^\mathrm{T}\!(k)
\begin{pmatrix}
Q_1+Q_\text{sync}&-Q_\text{sync}\\
-Q_\text{sync}&Q_2+Q_\text{sync}
\end{pmatrix}
\tilde{x}(k)
 +\tilde{u}^\mathrm{T}(k)
 \begin{pmatrix}
R_1&0\\
0&R_2
\end{pmatrix}
 \tilde{u}(k).
\end{align*}
Thus, the problem is in standard LQR form and can be solved with standard tools \cite{Anderson2007}.  The optimal stabilizing controller that minimizes \eqref{eqn:cost} has the structure
$u_1(k) = F_{11} x_1(k) + F_{12} x_2(k)$ and $u_2(k) = F_{21} x_1(k) + F_{22} x_2(k)$;
that is, agent~1 ($u_1(k)$) requires state information from agent~2 ($x_2(k)$), and vice versa.
Because of many-to-all communication, the wireless embedded system directly supports this (as well as any other possible) controller structure (\textbf{P3}).

Since the controller now runs on the node that is collocated with the physical system, local measurements and control inputs are not sent over the wireless network, and the local sampling time can be shorter than the update interval \Tupdate, while measurements from other agents are still received over the wireless network every \Tupdate.
Although the analysis in Sections~\ref{sec:stabAnalysis} and~\ref{sec:ctrl_mode_changes} can be extended to the synchronization setting, a formal stability proof is beyond the scope of this paper.
In general, stability is less critical here because of shorter update intervals in the local control loop.


\section{Experimental Evaluation}
\label{sec:eval}

This section uses measurements from a cyber-physical testbed with 20 wireless devices forming a three-hop network and several cart-pole systems to study the performance of the proposed wireless \cps design and validate the theoretical results.
Our experiments reveal the following key findings:
\begin{itemize}
 \item We can concurrently and safely stabilize two inverted pendulums over a three-hop low-power wireless network, either via a single remote controller or two separate remote controllers.
 \item Using the same wireless \cps design with a different control logic, we can reliably synchronize the movement of five cart-poles thanks to the support for arbitrary traffic patterns.
 \item We can dynamically change between well-defined modes that involve different control and communication requirements, without impairing closed-loop stability or control performance.
 \item Our system can stabilize an inverted pendulum at update intervals of \SIrange[range-units=single, range-phrase=-]{20}{50}{\milli\second}. At an update interval of \SI{20}{\ms}, it can stabilize an inverted pendulum despite \SI{75}{\percent} \iid Bernoulli losses and bursts of 40 consecutive losses. Larger update intervals decrease the control performance and the ability to tolerate message loss, but allow for saving communication resources.
 \item We can stabilize an inverted pendulum despite alternating between modes every \SIrange[range-units=single, range-phrase=--]{120}{240}{\milli\second}, which is significantly shorter than the minimum average dwell time stipulated by Theorem~\ref{thm:switch_system_high_level}.
 \item The measured jitter on the update interval and the end-to-end delay is less than $\SI{\pm25}{\micro\second}$, which validates our theoretical analysis of the worst-case jitter from \secref{sec:properties}.
\end{itemize}

\subsection{Cyber-physical Testbed and Performance Metrics}
\label{sec:cps_testbed}

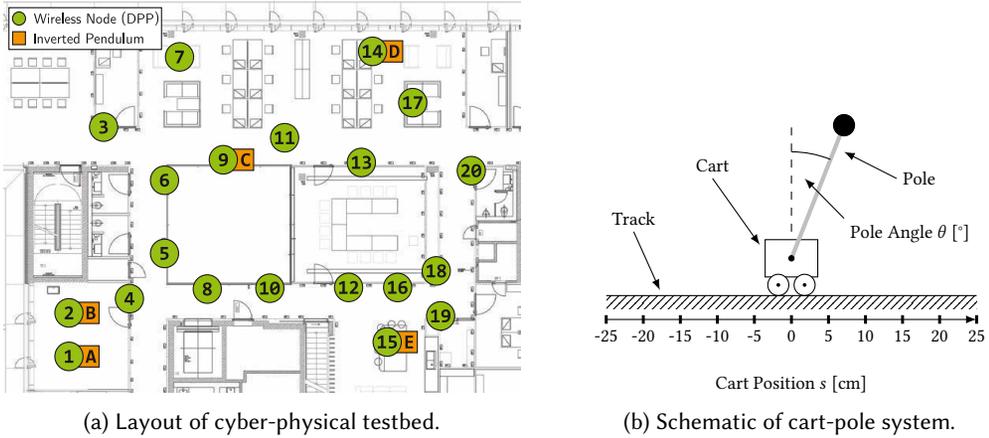
\begin{figure}
\begin{subfigure}[b]{0.495\linewidth}
	\centering
	\includegraphics[width=\linewidth]{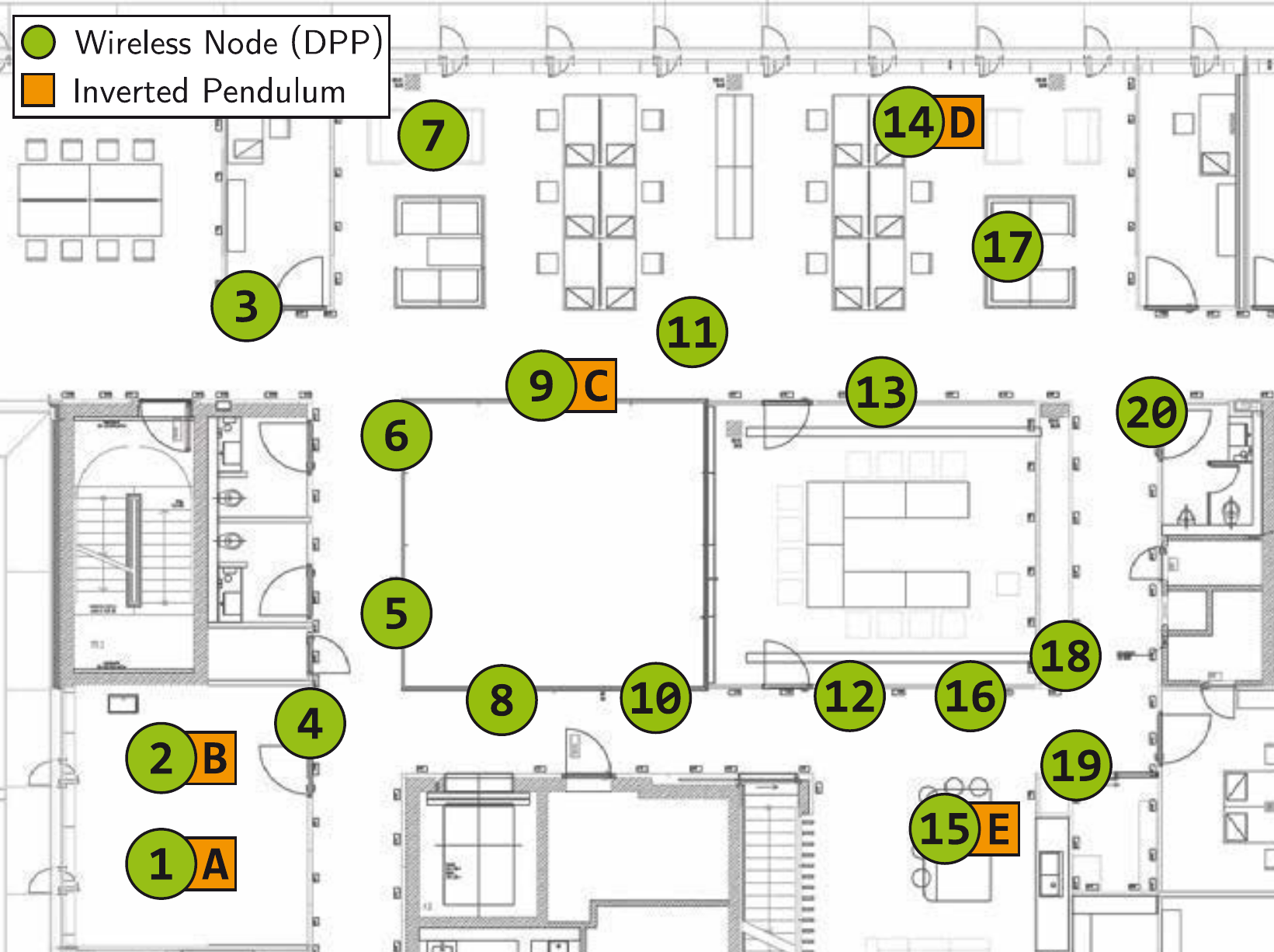}
	\caption{Layout of cyber-physical testbed.}
	\label{fig:testbed}
\end{subfigure}
\hfill
\begin{subfigure}[b]{0.495\linewidth}
	\centering
	\tikzsetnextfilename{cartPole}
	\begin{tikzpicture}[scale=0.7, every node/.style={scale=0.7}]]
	\tikzstyle{ground}=[fill,pattern=north east lines,draw=none,minimum width=5,minimum height=0.1]
	\tikzset{>=latex};
	\draw(0.5,0.35)rectangle(-0.5,-0.35);
	\draw[fill=white](-0.25,-0.5)circle(0.2);
	\draw[fill=white](0.25,-0.5)circle(0.2);
	\draw[ultra thick,lightgray](0,0)node(cart){}--(1,2.5);
	\draw[fill=black](1,2.5)circle(0.2)node(pole){};
	\draw[dashed] (0,0.5) -- (0,2.5)node(top){};
	\draw[fill=black](0,0)circle(0.05);
	\draw[fill=black](-0.25,-0.5)circle(0.025);
	\draw[fill=black](0.25,-0.5)circle(0.025);
	\node[ground,minimum width=200,anchor=north](floor)at(0,-0.7){};
	\draw(floor.north east)--(floor.north west);
	 \draw pic[" ",draw=black, -,  angle eccentricity=1.5,angle radius =
	  2cm] {angle = pole--cart--top};
	\draw[->](-1.5,1.5)node[above]{Cart} -- (-0.6,0.25);
	\draw[->] (2,1.5)node[right]{Pole} -- (1,2);
	\draw[->]([shift={(0,-0.2cm)}]floor.south west) -- node[midway,draw,rectangle,inner sep = 0, minimum width=0.1em, minimum height=0.4em,label=below:0,fill=black]{}
	node[pos=0.4,draw,rectangle,inner sep = 0, minimum width=0.1em, minimum height=0.4em,label=below:-5,fill=black]{}
	node[pos=0.3,draw,rectangle,inner sep = 0, minimum width=0.1em, minimum height=0.4em,label=below:-10,fill=black]{}
	node[pos=0.2,draw,rectangle,inner sep = 0, minimum width=0.1em, minimum height=0.4em,label=below:-15,fill=black]{}
	node[pos=0.1,draw,rectangle,inner sep = 0, minimum width=0.1em, minimum height=0.4em,label=below:-20,fill=black]{}
	node[pos=0,draw,rectangle,inner sep = 0, minimum width=0.1em, minimum height=0.4em,label=below:-25,fill=black]{}
	node[pos=0.6,draw,rectangle,inner sep = 0, minimum width=0.1em, minimum height=0.4em,label=below:5,fill=black]{}
	node[pos=0.7,draw,rectangle,inner sep = 0, minimum width=0.1em, minimum height=0.4em,label=below:10,fill=black]{}
	node[pos=0.8,draw,rectangle,inner sep = 0, minimum width=0.1em, minimum height=0.4em,label=below:15,fill=black]{}
	node[pos=0.9,draw,rectangle,inner sep = 0, minimum width=0.1em, minimum height=0.4em,label=below:20,fill=black]{}
	node[pos=1,draw,rectangle,inner sep = 0, minimum width=0.1em, minimum height=0.4em,label=below:25,fill=black]{}
	([shift={(0,-0.2cm)}]floor.south east);
	\draw[->] (-3,0.5)node[above]{Track} -- ([shift={(-2.5,0.1cm)}]floor.north);
	\node[below = 2.5em of floor,inner sep = 0]{Cart Position $s$ [\si{\centi\meter}]};
	\draw[->](1,0.75)node[below right]{Pole Angle $\theta$ [\si{\degree}]} -- (0.25,1.5);
	\end{tikzpicture}
	\caption{Schematic of cart-pole system.}
	\label{fig:pendulum}
\end{subfigure}
\caption{Cyber-physical testbed with 20 \dpp nodes forming a three-hop wireless network and five cart-pole systems, two real ones connected to nodes 1 and 2 and three simulated ones running on nodes 9, 14, and 15.}
\label{fig:testbed_overall}
\end{figure}

Realistic cyber-physical testbeds are essential for the validation and evaluation of \cps solutions~\cite{Lu2016a,Baumann2018}.
With the goal of capturing the requirements of a large class of emerging \cps applications~\cite{Akerberg2011,Luvisotto2017}, we develop the wireless cyber-physical testbed shown in \figref{fig:testbed_overall}.
The testbed is deployed in an office building across an area of \SI{15}{\m} by \SI{20}{\m}, as illustrated in \figref{fig:testbed}.
It consists of 20 \dpp nodes (see \secref{subsec:embeddedHardware}), two real physical systems (A and B), and three simulated physical systems (C, D, and E).
All \dpp nodes transmit at \SI{10}{\dBm}, which results in a network diameter of three hops.
The wireless signals need to penetrate various types of walls, from glass to reinforced concrete, and are exposed to different sources of interference from other electronics and human activity.

We use \emph{cart-pole systems} as physical systems.
As shown in \figref{fig:pendulum}, a cart-pole system consists of a cart that can move horizontally on a track and a pole attached to it via a revolute joint.
The cart is equipped with a DC motor that can be controlled by applying a voltage to influence the speed and the direction of the cart.
Moving the cart exerts a force on the pole and thus influences the pole angle~$\theta$.
In this way, the pole can be stabilized in an upright position around $\theta=\SI{0}{\degree}$, which represents an unstable equilibrium and is called the \emph{inverted pendulum}.
The inverted pendulum has fast dynamics, which are typical of real-world mechanical systems~\cite{Boubaker2012}, and requires feedback with update intervals of tens of milliseconds.

For small deviations from the equilibrium (\ie $\sin(\theta) \approx \theta$), the inverted pendulum can be well approximated as an \lti system.
The state $x(k)$ of the system consists of four variables.
Two of them, the pole angle~$\theta(k)$ and the cart position~$s(k)$, are directly measured by angle sensors.
Their derivatives, the angular velocity~$\dot{\theta}(k)$ and the cart velocity~$\dot{s}(k)$, are estimated using finite differences and low-pass filtering.
The voltage applied to the motor is the control input $u(k)$.
In this way, the APs of \dpp nodes 1 and 2 interact with the two real pendulums A and B, while the APs of nodes 9, 14, and 15 run a simulation model of the inverted pendulum, labelled as C, D, and E in \figref{fig:testbed}.

The cart-pole system has a few constraints.
Control inputs are capped at \SI{\pm10}{\V}.
The track has a length of \SI{\pm25}{\cm} from the center (see \figref{fig:pendulum}).
Surpassing the track limits ends an experiment.
Before each experiment, we move the carts to the center and the poles in the upright position; then the controller takes over.
\secref{sec:controller_implementation} details the implementation of the controllers for multi-hop stabilization and multi-hop synchronization, following the design outlined in \mbox{Sections~\ref{sec:ctrlDesign} and~\ref{sec:sync}}.

We measure the control performance in terms of pole angle, cart position, and control input.
Furthermore, we measure the nodes' radio duty cycle in software, which can be considered the communication costs of feedback control, and record when a message is lost over wireless.

\db{The same setup has also been used in larger environments, see for instance~\cite{mager2019b} (\url{https://youtu.be/AtULmfGkVCE}).}
\mz{For now, I discuss this in the conclusions. We can discuss if this is what we want.}
  
\subsection{Multi-hop Stabilization}
\label{sec:multihop_stabilization}

In our first experiment, we study the feasibility and performance of fast feedback control over a real multi-hop low-power wireless network, thereby validating the theoretical results from \secref{sec:stabAnalysis}.

%
%
\begin{figure}
\begin{subfigure}[t]{0.495\linewidth}
    \centering
    \includegraphics[width=\textwidth]{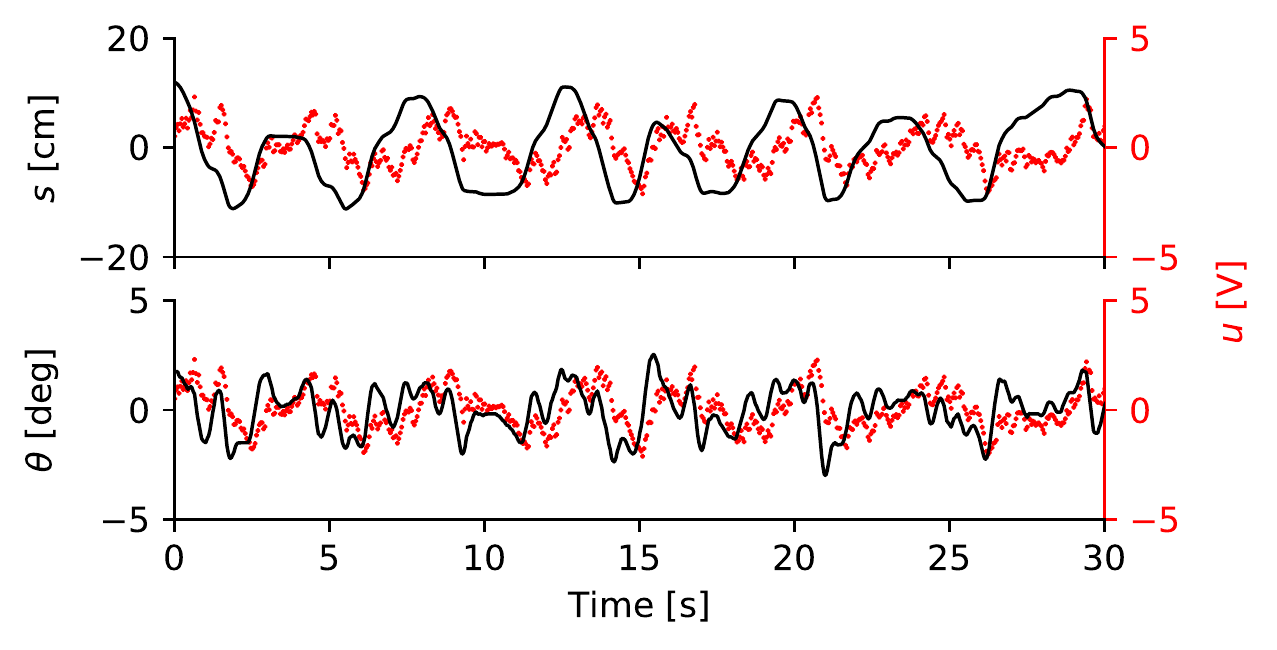}
    \subcaption{Pendulum A.}
    \label{fig:stabilization_pendA}
\end{subfigure}
\begin{subfigure}[t]{0.495\linewidth}
    \centering
    \includegraphics[width=\textwidth]{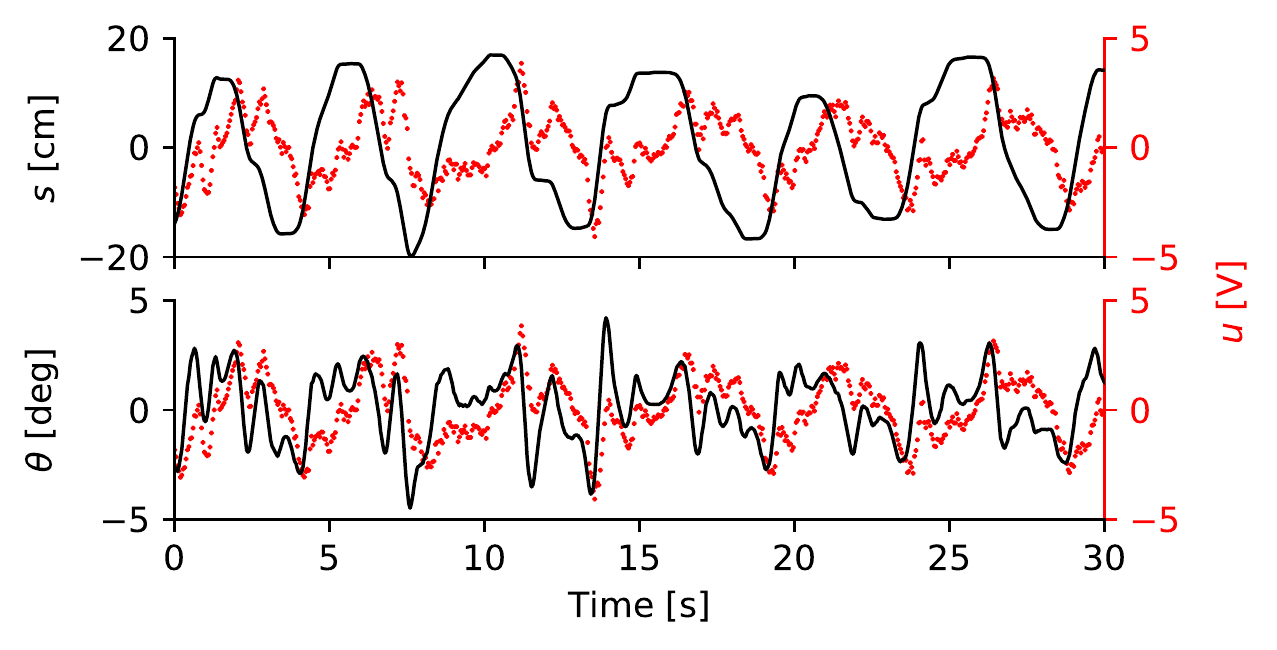}
    \subcaption{Pendulum B.}
    \label{fig:stabilization_pendB}
\end{subfigure}
\caption{Cart position $s$, pole angle $\theta$, and control input $u$ of two real inverted pendulums when concurrently stabilizing them over the same multi-hop wireless network at an update interval of 45\ms. \capt{The cart position and the pole angle always stay within safe regimes, and less than half of the possible control input is needed.}}
\label{fig:stabilization}
\end{figure}

\fakepar{Setup}
We use two controllers running on nodes 14 and 15 to stabilize the two real pendulums A and B at $\theta = \SI{0}{\degree}$ and $s = \SI{0}{\cm}$.
Hence, there are two independent control loops sharing the same wireless network, and it takes six hops to close each loop.
We configure the wireless embedded system and the controllers for an update interval of $\Tupdate = \SI{45}{\ms}$, and according to~\textbf{P2} (and confirmed by our measurements discussed below) we expect a message delivery rate of at least \SI{99.9}{\percent}.
For these settings, we can use Theorems~\ref{thm:MSSourSystem} and~\ref{thm:MSSourSysNoise} to \emph{formally prove} stability of the overall system.

\fakepar{Results}
Our experimental results \emph{empirically validate} the system design in the tested scenario: We can safely stabilize both pendulums over the three-hop wireless network.
\figref{fig:stabilization} shows a characteristic \SI{30}{\second} trace of the two pendulums.
Due to differences in the physical properties of the two pendulums, cart position, pole angle, and control input oscillate differently, but always stay within safe regimes.
For example, the carts never come very close to either end of their track, and less than half of the possible control input is applied.
Not a single message was lost in this experiment, which demonstrates the reliability of our wireless embedded system design.

\begin{figure}[t]
\begin{subfigure}{0.495\linewidth}
	\centering
	\includegraphics[width=\linewidth]{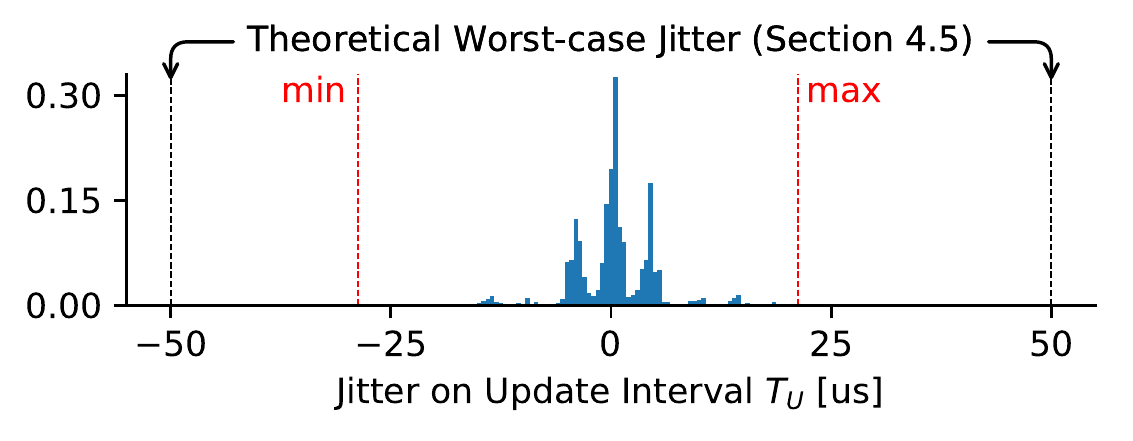}
	\caption{Jitter on update interval.}
	\label{fig:update_interval_jitter}
\end{subfigure}
\begin{subfigure}{0.495\linewidth}
	\centering
	\includegraphics[width=\linewidth]{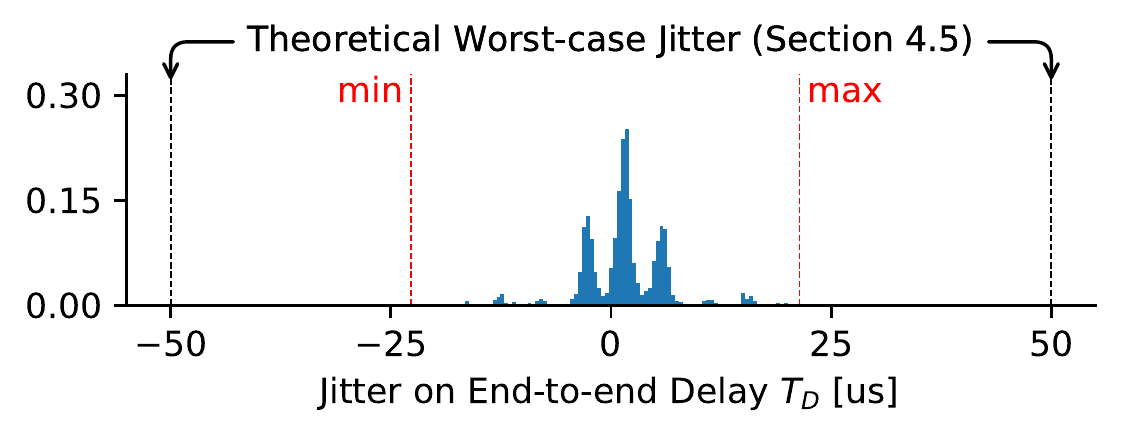}
	\caption{Jitter on end-to-end delay.}
	\label{fig:end_to_end_delay_jitter}
\end{subfigure}
\caption{Distributions of the measured jitter on the update interval \Tupdate and the end-to-end delay \Tdelay. \capt{The experimental measurements are well within the theoretical worst-case bounds determined in \secref{sec:properties}.}}
\label{fig:interval_jitter}
\end{figure}

During the same experiment, we use a logic analyzer to measure the update interval \Tupdate and the end-to-end delay \Tdelay (see \figref{fig:schedule}).
\figref{fig:interval_jitter} shows the distributions of the measured jitter on \Tupdate and \Tdelay.
We see that the empirical results are well within the theoretical worst-case bounds of $\pm\SI{50}{\micro\second}$, which validates our jitter analysis from \secref{sec:properties} and assumptions made in \secref{sec:control}.
Indeed, this jitter is several orders\db{we are three orders of magnitude faster than the TiSCH network, so why writing one here? I would probably either write several or three.} of magnitude smaller than the jitter of conventional approaches based on routing and per-link scheduling.
For example, Schindler \etal report that the \emph{communication} jitter alone (\ie neglecting time-varying processing delays, which would contribute to the \emph{end-to-end} jitter) is at least $\pm\SI{23}{\milli\second}$ in a \emph{single-hop} 6TiSCH network, which is an advancement of WirelessHART~\cite{Schindler2017}.
Unlike our approach, such jitter cannot be neglected as it is on par with the dynamics of the physical systems, complicating control design and stability analysis.

In a further experiment, we demonstrate that our design is indeed independent of the network topology (see \secref{sec:wireless_protocol}).
This independence allows us, for instance, to carry around the controller while balancing the pendulums, without any deterioration of the control performance.\footnote{A video of this experiment can be found at  \url{https://youtu.be/19xPHjnobkY}.}

\subsection{Multi-hop Synchronization}
\label{sec:synchronization}

To assess its versatility, we apply the same \cps design to the distributed control task from \secref{sec:sync}.

\fakepar{Setup}
We use the two real pendulums A and B and the three simulated pendulums C, D, and E.
The goal is to synchronize the cart positions of the five pendulums over the multi-hop wireless network, while each pendulum is stabilized  by a local control loop.
This scenario is similar to drone swarm coordination, where each drone stabilizes its flight locally, but exchanges its position with all other drones to keep a desired swarm formation~\cite{Preiss2017}.
In our experiment, stabilization runs with $\Tupdate = \SI{10}{\milli\second}$, and nodes 1, 2, 9, 14, and 15 exchange their current cart positions every \SI{50}{\milli\second}.

\begin{figure*}[!bt]
	\centering
	\includegraphics[width=\linewidth]{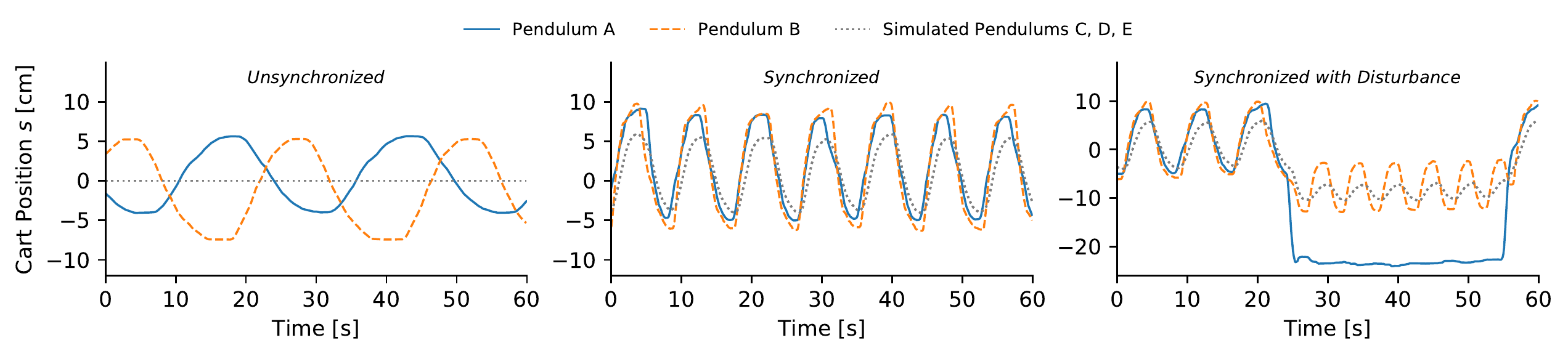}
	\caption{Cart position of five cart-pole systems over time as they locally stabilize their pole at an update interval of \SI{10}{\ms} and synchronize the cart positions (middle and right plot) over the multi-hop low-power wireless network at an update interval of 50\ms. \capt{With synchronization enabled, all three carts move in concert and even try to mimic the temporary disturbance of pendulum A (right plot).}}
	\label{fig:synchronization_position_time_plot}
\end{figure*}

\fakepar{Results}
The left plot in \figref{fig:synchronization_position_time_plot} shows the cart positions over time without synchronization.
We see that the carts of the real pendulums move with different amplitude, phase, and frequency due to slight differences in their physics and imperfect measurements.
The simulated pendulums, instead, are perfectly balanced and behave deterministically as they all start in the same initial state.

In the middle plot of \figref{fig:synchronization_position_time_plot}, we can observe the behavior of the pendulums when they synchronize their cart positions over the multi-hop wireless network.
Now, all five carts move in concert.
The movements are not perfectly aligned because, besides the synchronization, each cart also needs to locally stabilize its pole at $\theta = \SI{0}{\degree}$ and $s = \SI{0}{\cm}$.
Since no message is lost during the experiment, the simulated pendulums all receive the same state information and, therefore, show identical behavior.

This effect can also be seen in our third experiment, shown in the right plot of \figref{fig:synchronization_position_time_plot}, where we hold pendulum A for some time at $s = \SI{-20}{\cm}$.
The other pendulums now have two conflicting control goals: stabilization at $s=\SI{0}{\cm}$ and $\theta=\SI{0}{\degree}$, as well as synchronization while one pendulum is fixed at about $s = \SI{-20}{\cm}$.
As a result, they all move towards this position and oscillate between $s = 0$ and $s = \SI{-20}{\cm}$.
Clearly, this experiment demonstrates that the cart-pole systems influence each other, which is enabled by the many-to-all communication over the multi-hop wireless network.

\begin{table}[tb]
\addtolength{\tabcolsep}{2.6pt}
\caption{Modes considered in the experiment of \secref{sec:eval_mode_changes} including the update intervals used. In mode~2, the two real pendulums are remotely stabilized by one controller (ctrl) running on node~13. A second controller running on node~18 is used in mode~3 so that each real pendulum has its own controller.}
\begin{tabular}{ccccc}
\toprule
\multirow{2}[3]{*}{Mode} & \multicolumn{2}{c}{Real pendulums A and B} & \multicolumn{2}{c}{Simulated pendulums C, D, and E}\\
\cmidrule(lr){2-3} \cmidrule(lr){4-5}
& Stabilization & Synchronization & Stabilization & Synchronization \\
\midrule
1 & Local @ \SI{10}{\milli\second} & -- & Local @ \SI{10}{\milli\second} & -- \\
2 & Remote (1 ctrl) @ \SI{40}{\milli\second} & -- & Local @ \SI{10}{\milli\second} & -- \\
3 & Remote (2 ctrl) @ \SI{45}{\milli\second} & -- & Local @ \SI{10}{\milli\second} & -- \\
4 & Local @ \SI{10}{\milli\second} & \SI{50}{\milli\second} & Local @ \SI{10}{\milli\second} & -- \\
5 & Local @ \SI{10}{\milli\second} & \SI{50}{\milli\second} & Local @ \SI{10}{\milli\second} & \SI{50}{\milli\second} \\
\bottomrule
\end{tabular}
\label{tbl:mode_spec}
\addtolength{\tabcolsep}{-2.6pt}
\end{table}


\begin{sidewaysfigure}
	\includegraphics[width=1\linewidth]{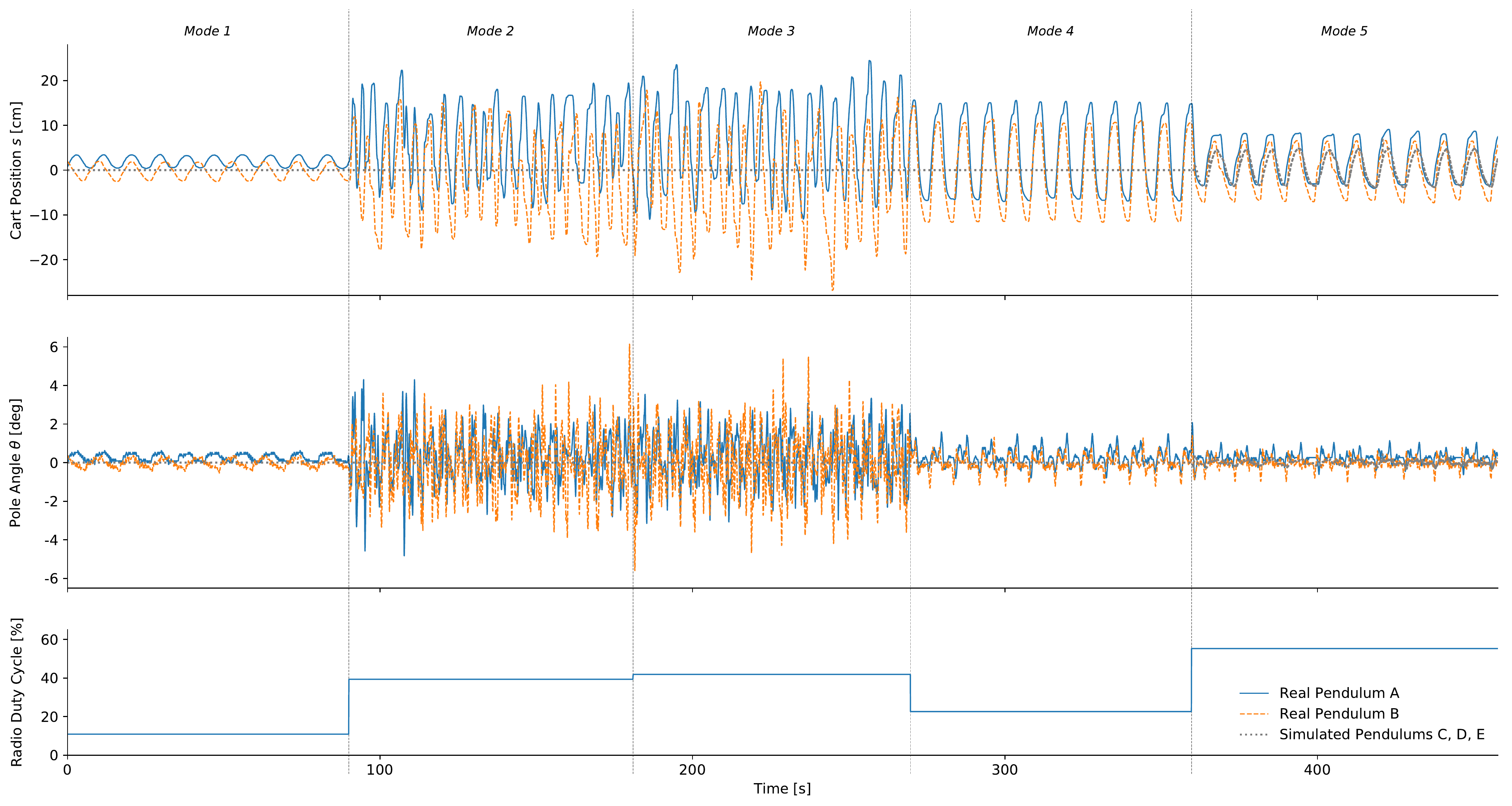}
	\caption{Control performance of all five pendulums and per-mode average radio duty cycle of the node attached to pendulum A throughout the example sequence $1 \rightarrow 2 \rightarrow 3 \rightarrow 4 \rightarrow 5$ of mode changes (see \tabref{tbl:mode_spec}). \capt{Our system can safely change between different modes without sacrificing closed-loop stability.}}
	\label{fig:mode_changes_long_exp}
\end{sidewaysfigure}

\subsection{Mode Changes}
\label{sec:eval_mode_changes}

Next, we test the ability of our system to dynamically change between different modes at runtime without affecting stability, thereby validating the theoretical results from \secref{sec:ctrl_mode_changes}.

\fakepar{Setup}
We consider the five modes in \tabref{tbl:mode_spec}.
They differ in the number and the way the cart-poles in our testbed\footnote{Due to logistic constraints, we slighly adapt the testbed layout for this experiment (see \figref{fig:floorplan_mode_changes} in supplementary material).} are stabilized and/or synchronized, which implies different application tasks being executed on certain nodes as well as different traffic loads, traffic patterns, and update intervals.
We compute the corresponding schedules offline (see \secref{subsec:scheduling}), and distribute them to the nodes before the experiment begins.
At runtime, we manually trigger a mode change at the host every \SI{90}{\second}, which becomes effective $r=5$ communication rounds later using the mode-change protocol from \secref{sec:design_mode_switches}.
For these settings, we can use Theorem~\ref{thm:switch_system_high_level} to \emph{formally prove stability}: The average dwell time over any given interval (corresponding to \SI{90}{\second}) is at least 2000 discrete time steps (dwell time in mode 3), which exceeds the required minimum average dwell time of $\tau_\mathrm{a}^* = 289$ (\SI{11.56}{\second} in mode 2 and \SI{13.005}{\second} in mode 3) discrete time steps.


\fakepar{Results}
\figref{fig:mode_changes_long_exp} shows cart position and pendulum angle of all five cart-poles over time for the example sequence $1 \rightarrow 2 \rightarrow 3 \rightarrow 4 \rightarrow 5$ of mode changes; the bottommost plot shows for each mode the average radio duty cycle of the node at pendulum A.
Our results empirically validate the theoretical results for the tested scenario: The system remains stable despite the mode changes.
The general behavior of the pendulums with and without synchronization is similar to the previous experiment.
Interestingly, however, when synchronizing all five carts in mode 5, the two real pendulums exhibit less variation in terms of cart position and pendulum angle compared with mode 4: The idealized simulation model underlying pendulums C, D, and E influence the behavior of the real pendulums A and B, and vice versa.
Looking at the duty cycle, we notice that each mode incurs a different communication cost.
Even in mode 1, where the application tasks exchange no messages over the network, the duty cycle is \SI{10}{\percent}.
This is due to the beacon sent at the beginning of communication rounds, scheduled with a period of \SI{25}{\milli\second} also in mode 1 to quickly react to mode-change events.

\subsection{Impact of Update Interval}
\label{sec:update_interval}

The following experiment takes a closer look at the impact of different update intervals (and hence different end-to-end delays) on the control performance and the associated communication costs.

\fakepar{Setup}
To minimize effects that we cannot control, such as external interference, we use two nodes close to each other: Real pendulum A attached to node 1 is stabilized via a remote controller running on node 2 (\cf \figref{fig:testbed_overall}).
We test different update intervals in consecutive runs.
Starting with the smallest update interval of 20\ms that the wireless embedded system can support in this scenario, we increase the update interval in steps of 10\ms until stabilization is no longer possible.

\begin{figure}
\hspace*{\fill}
\begin{subfigure}[t]{0.2425\linewidth}
    \centering
    \includegraphics[width=1\textwidth]{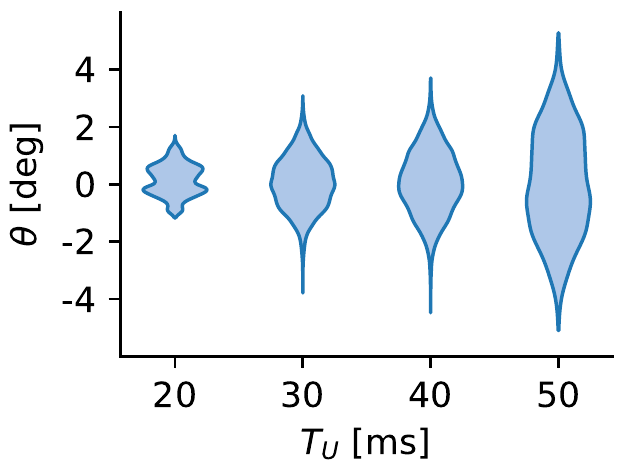}
    \subcaption{Pole angle.}
\end{subfigure}
\hfill
\begin{subfigure}[t]{0.2425\linewidth}
    \centering
    \includegraphics[width=1\textwidth]{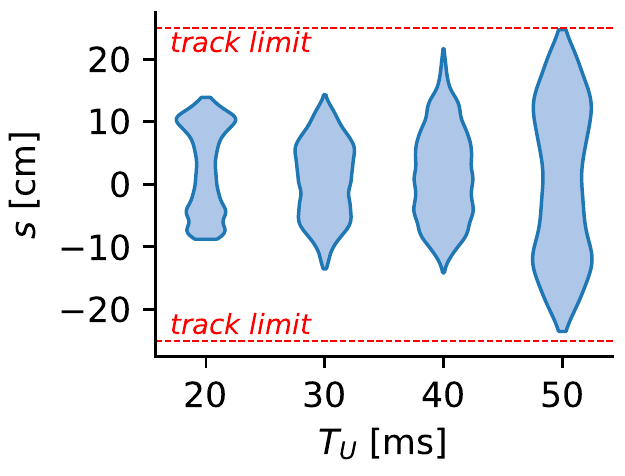}
    \subcaption{Cart position.}
\end{subfigure}
\hfill
\begin{subfigure}[t]{0.2425\linewidth}
    \centering
    \includegraphics[width=1\textwidth]{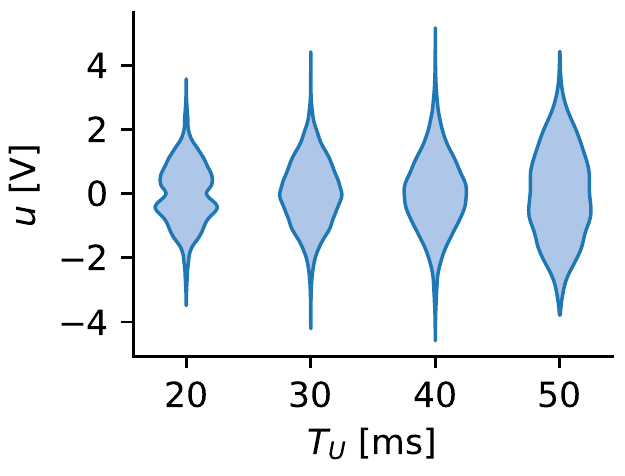}
    \subcaption{Control input.}
\end{subfigure}
\hfill
\begin{subfigure}[t]{0.2425\linewidth}
    \centering
    \includegraphics[width=1\textwidth]{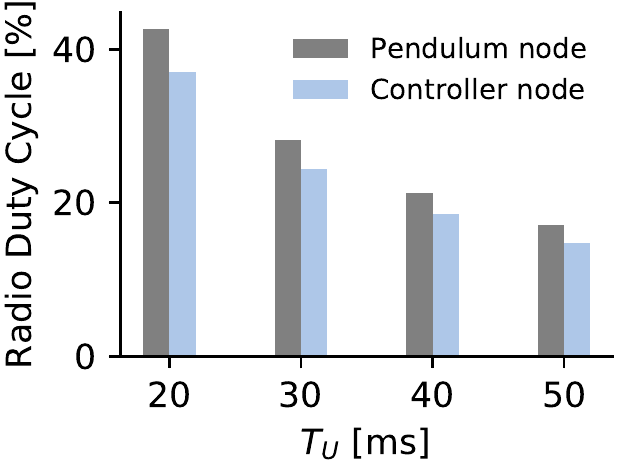}
    \subcaption{Radio duty cycle.}
\end{subfigure}
\hspace*{\fill}
\caption{Distribution of control performance metrics and average radio duty cycle when stabilizing an inverted pendulum over low-power wireless at different update intervals. The plots show the distribution of the respective variable for fixed update intervals \Tupdate. \capt{A larger update interval leads to larger pole angles and more movement of the cart, but also significantly reduces the communication costs for feedback control.}}
\label{fig:timescale}
\end{figure}

\fakepar{Results}
\figref{fig:timescale} shows control performance and radio duty cycle for different update intervals based on more than \num{12500} data points.
We see that a longer update interval causes larger pole angles and more movement of the cart.
Indeed, the total distance the cart moves during an experiment increases from \SI{3.40}{\m} for \SI{20}{\milli\second} to \SI{9.78}{\m} for \SI{50}{\milli\second}.
This is consistent with the wider distribution of the control input for larger update intervals.
At the same time, the radio duty cycle decreases from 40\percent for \SI{20}{\milli\second} to 15\percent for \SI{50}{\milli\second}.
Hence, there is a trade-off between control performance and the associated communication costs, which may be exploited based on the application requirements.

\subsection{Impact of Dwell Time}
\label{sec:eval_dwell_time}

Unlike the experiment in \secref{sec:eval_mode_changes}, we now evaluate control performance and stability when mode changes occur faster than the minimum average dwell time stipulated by Theorem~\ref{thm:switch_system_high_level} allows.

\fakepar{Setup}
We consider the setup from the last experiment and two modes: remote stabilization with an update interval of \SI{30}{\milli\second} and \SI{40}{\milli\second}.
We configure the mode-change protocol such that the new mode becomes effective $r=\{25,50,100\}$ communication rounds after the first request, and let the host request the next mode change already one round later.
Different from the experiment in \secref{sec:eval_mode_changes}, with these settings stability is \emph{not} guaranteed according to Theorem~\ref{thm:switch_system_high_level} as the system changes modes \emph{significantly faster} than the required minimum average dwell time of $\tau_\mathrm{a}^* = 272$.
Note that $r=25$ results in the shortest dwell time at which stabilization is possible in our setting.


\begin{figure}
\hspace*{\fill}
\begin{subfigure}[t]{0.2425\linewidth}
    \centering
    \includegraphics[width=1\textwidth]{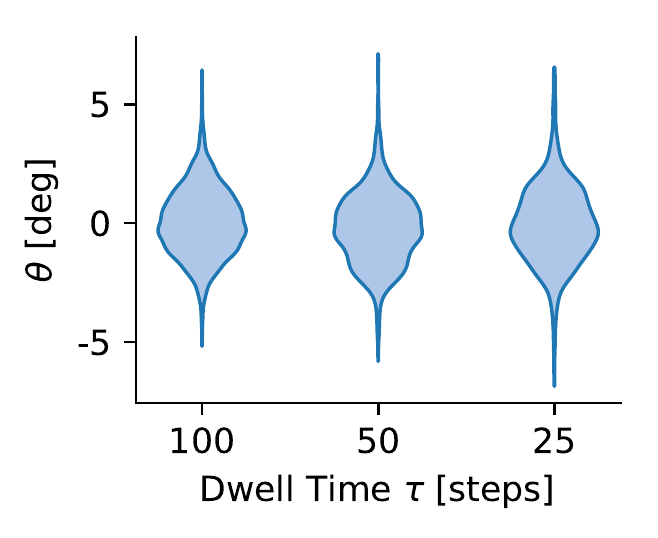}
    \subcaption{Pole angle.}
\end{subfigure}
\hfill
\begin{subfigure}[t]{0.2425\linewidth}
    \centering
    \includegraphics[width=1\textwidth]{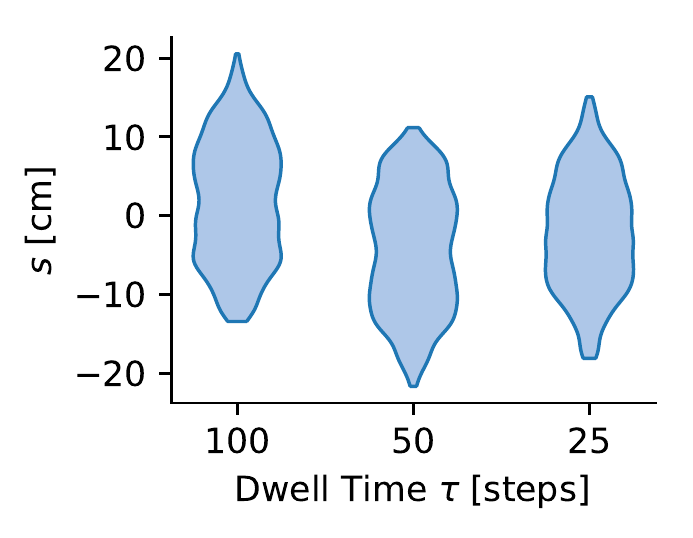}
    \subcaption{Cart position.}
\end{subfigure}
\hfill
\begin{subfigure}[t]{0.2425\linewidth}
    \centering
    \includegraphics[width=1\textwidth]{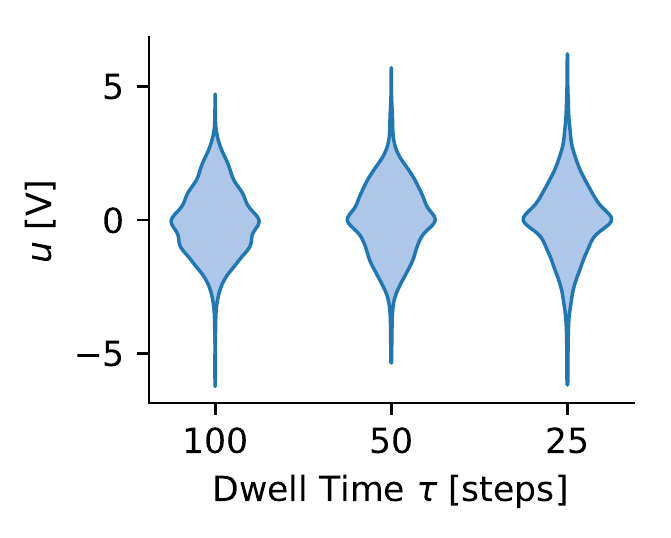}
    \subcaption{Control input.}
\end{subfigure}
\hspace*{\fill}
\caption{Distribution of control performance when stabilizing an inverted pendulum while changing modes with different dwell times. \capt{The control performance shows no noticable difference across the tested dwell times.}}
\label{fig:quick_switches}
\end{figure}

\fakepar{Results}
\figref{fig:quick_switches} shows that pole angle and cart position always stay in a safe regime and rarely more than half of the possible input voltage is used.
A difference in control performance is hardly visible across the three dwell times, suggesting that control performance is nearly independent of the dwell time when stabilization is possible.
As the system is stable even for dwell times much smaller than the average dwell time requested by Theorem~\ref{thm:switch_system_high_level}, we can conclude that it is indeed safe to neglect the dead time (\textbf{P4}) in the stability analysis from \secref{sec:ctrl_mode_changes}.


\subsection{Resilience to Message Loss}
\label{sec:resilience} 

Finally, we look at how control performance is affected by significant message loss over wireless.

\fakepar{Setup}
We use again the two-node setup and fix the update interval at \SI{20}{\ms}.
We let both nodes intentionally drop messages in two different ways.
In a first experiment, the two nodes independently drop a received message according to a Bernoulli process with given failure probability.
We test three failure probabilities in different runs: 15\percent, 45\percent, and 75\percent.
In a second experiment, the two nodes drop a certain number of consecutive messages every \SI{10}{\s}, namely between 10 and 40 messages in different runs.
This artificially violates property \textbf{P2} of the wireless embedded system, yet allows us to evaluate the robustness of our control design to unexpected conditions.

\begin{figure}
\hspace*{\fill}
\begin{subfigure}[t]{0.2425\linewidth}
    \centering
    \includegraphics[width=\textwidth]{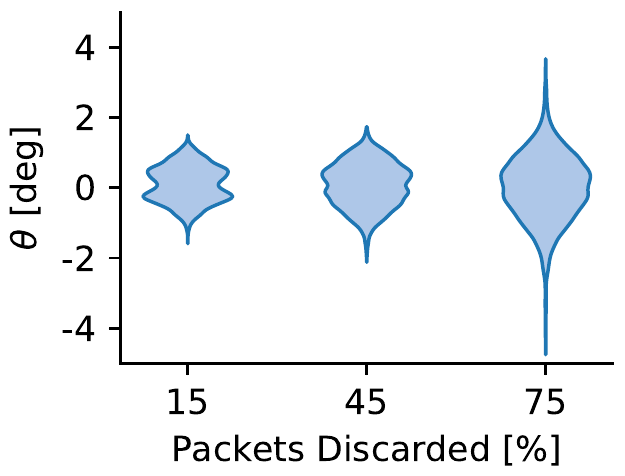}
    \subcaption{Pole angle.}
    \label{fig:drops_angle}
\end{subfigure}
\hfill
\begin{subfigure}[t]{0.2425\linewidth}
    \centering
    \includegraphics[width=\textwidth]{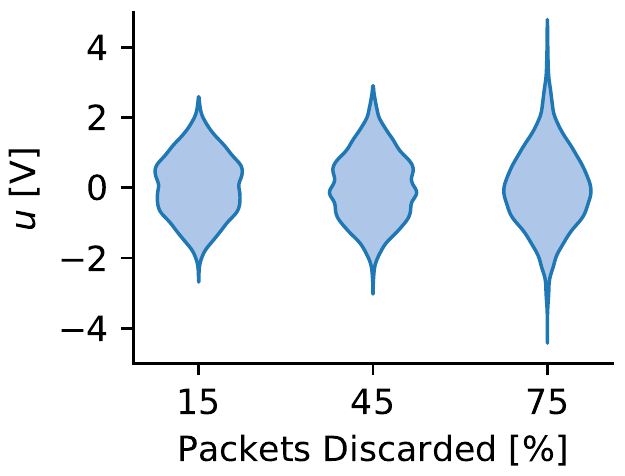}
    \subcaption{Control input.}
    \label{fig:drops_voltage}
\end{subfigure}
\hfill
\begin{subfigure}[t]{0.2425\linewidth}
    \centering
    \includegraphics[width=\textwidth]{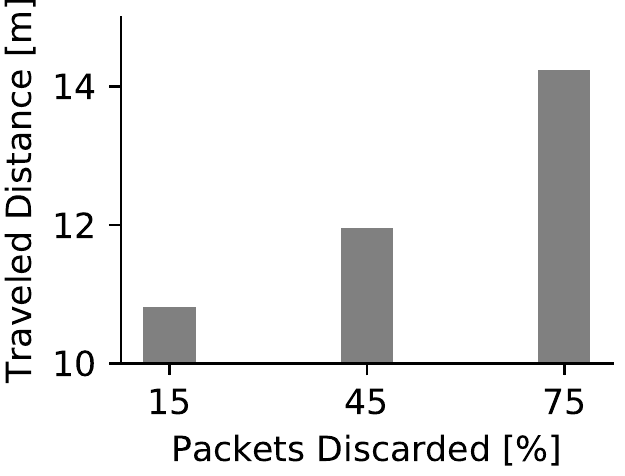}
    \subcaption{Traveled distance.}
    \label{fig:drops_distance}
\end{subfigure}
\hspace*{\fill}
\par\medskip
\begin{subfigure}[t]{1\linewidth}
    \centering
    \includegraphics[width=\textwidth]{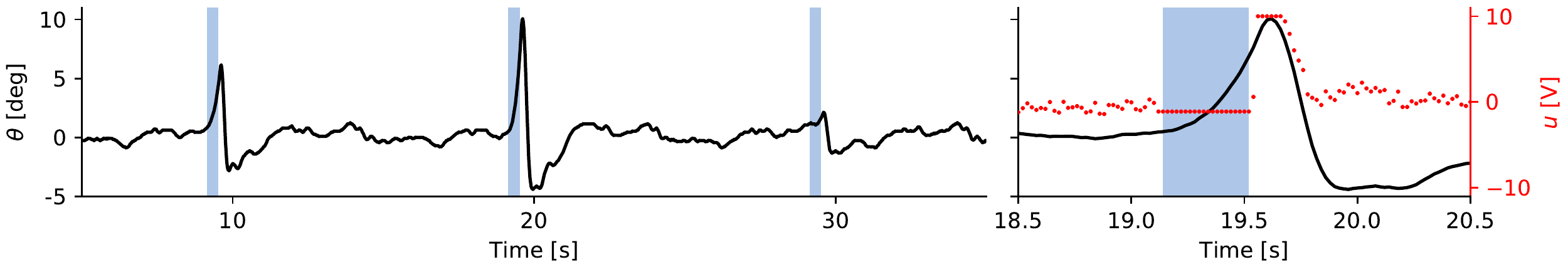}
    \subcaption{Pole angle over time for three artificial bursts of 40 consecutively lost messages every \SI{10}{\s} (shaded areas). The plot to the right zoomes into the second burst phase.}
    \label{fig:burst}
\end{subfigure}
\caption{Control performance when stabilizing a pendulum over low-power wireless under artificially injected message loss, for \iid Bernoulli losses in (a), (b), and (c) and for bursts of multiple consecutive losses in (d). \capt{Depending on the update interval, the pendulum can be stabilized despite significant and bursty message loss.}}
\label{fig:drops}
\end{figure}

\fakepar{Results}
Figures~\ref{fig:drops_angle} and \ref{fig:drops_voltage} show the distributions of the pole angle and the control input for the three failure probabilities.
We see that the control performance decreases at higher loss rates, but the pendulum can be stabilized even at a loss rate of 75\%.
One reason for this is the short update interval.
For example, losing 50\percent of the messages at an update interval of 20\ms is comparable to an update interval of 40\ms without message loss, which is sufficient to stabilize the pendulum.
With a longer update interval, the system would not be able to tolerate such high message loss.

\figref{fig:burst} plots the pole angle over time for a burst length of 40 consecutively lost messages, with the right plot zooming into the second burst phase.
No control inputs are received during a burst, and depending on the state of the pendulum and the control input right before a burst, the impact of a burst can be very different as visible in \figref{fig:burst}.
The magnified plot shows that the pole angle diverges from around \SI{0}{\degree} with increasing speed.
When the burst ends, the control input rises to its maximum value of \SI{10}{\V} to bring the pendulum back to a non-critical state, which usually takes \SIrange[range-units=single, range-phrase=--]{1}{2}{\s}.
These results show that while property \textbf{P2} of our wireless embedded system design significantly simplifies control design and analysis, the overall system remains stable even if \textbf{P2} is dramatically violated, which is nevertheless very unlikely as demonstrated in prior work~\cite{Zimmerling2013,Karschau2018}.


\section{Concluding Remarks}
\label{sec:ending}
We have presented a \cps design that enables, for the first time, fast feedback control with stability guarantees and mode changes over wireless multi-hop networks at update intervals of \SIrange[range-units=single, range-phrase=-]{20}{50}{\milli\second}.
Existing solutions for feedback control over real wireless networks are either limited to single-hop scenarios or physical systems with slow dynamics, where update intervals on the order of seconds are sufficient.
Using a novel co-design approach, we tame communication imperfections and account for the resulting key properties of the wireless embedded system in the control design.
This has allowed us to formally prove closed-loop stability of the entire \cps in the presence of noise and mode changes.
Experiments on a cyber-physical testbed with multiple physical systems further demonstrate the applicability, versatility, and robustness of our design.
We thus maintain that our work represents an important stepping stone toward realizing the \cps vision.

We have deployed our system and testbed in other locations.
This includes a public demonstration at CPS-IoT Week 2019~\cite{mager2019b} featuring three real inverted pendulums spread out across an indoor space about twice as large as the one used for the experiments in this article.
Remote stabilization and synchronization worked reliably over a three-hop network despite the larger physical space, the presence of hundreds of people, and interference from other wireless equipment.\footnote{A video of the demonstration can be found at \url{https://youtu.be/AtULmfGkVCE}.}

In terms of scalability, our current implementation satisfies the needs of many \cps applications relying on tens of devices forming a network with a diameter of up to three hops~\cite{Akerberg2011,Luvisotto2017}.
Indeed, in our system, the minimum update interval and end-to-end delay are independent of the number of devices, but increase linearly with the traffic load and the network diameter.
This limits the number of physical systems and the extent of the deployment area that can be supported for a given control task.
Specifically, our current implementation supports remote stabilization and synchronization of up to three and five inverted pendulums, respectively, over a three-hop network at an update interval of \SI{50}{\milli\second}.
One possibility to surpass these limits is to switch to a faster low-power wireless physical layer.
For instance, Al Nahas \etal have recently demonstrated an implementation of a network-flooding primitive similar to Glossy on Bluetooth Low Energy 5 radios that can operate at a transmit bitrate of up to 2\,Mbit/s, which is 8$\times$ faster than the radios we use~\cite{AlNahas2019}.
With this, we could seamlessly support more physical systems and deeper networks with more hops.

Although 5G standardization is still ongoing and initial field trials have just begun, we briefly comment on similarities and differences between 5G and our work.
5G requires the deployment of dedicated infrastructure (base stations) and devices operating in licensed spectrum.
By contrast, our approach is infrastructure-less and targets commodity, off-the-shelf devices that are available today and operate in unlicensed spectrum.
As a result, our approach incurs lower costs compared with 5G and allows users to remain independent from network operators, with full control over their hardware and data.
Network slicing and ultra-reliable low-lateny communication (URLLC) are two key ingredients of 5G~\cite{Wollschlaeger2017}.
Network slicing entails the allocation of network resources (\eg bandwidth) to an application or service.
Using TTW, which supports the scheduling of multiple CPS applications, we essentially perform static network slicing in a wireless multi-hop network, while the mode switches allow us to support a predefined set of slicing configurations.
URLLC targets packet loss rates of $10^{-9}$ and packet latencies of \SI{1}{\milli\second} between network interfaces.
While short network latencies are indeed a prerequisite to control fast physical systems, our work demonstrates how to achieve short end-to-end latencies with bounded, negligible jitter among application tasks and that closed-loop stability guarantees do not necessarily require close-to-zero packet loss.


\begin{acks}

We thank Harsoveet Singh and  Felix Grimminger for
help with the testbed, and the TEC group at ETH Zurich
for making the design of the DPP platform available to the
public.
This work was supported in part by the German Research Foundation (DFG) within the Cluster of Excellence CFAED (grant 
EXC 1056), SPP 1914 (grants ZI 1635/1-1 and TR 1433/1-1), and 
the Emmy Noether project NextIoT (grant ZI 1635/2-1); the Cyber Valley Initiative; and the Max Planck Society.

\end{acks}

\section*{Author Contributions}
The electronic appendix, which can be accessed in the ACM Digital Library, contains a detailed description of the authors' individual contributions to the work presented in this article.

\bibliographystyle{ACM-Reference-Format}
\bibliography{abbrv,ref}

\clearpage

\appendix

\section{Supplementary Materials}

\subsection{Proof of Theorem~\ref{thm:MSSourSysNoise}}
\label{sec:noise}
We now consider mean-square stability of the system described in~\eqref{eqn:gen_sys_noise} without setting $\epsilon(k)=0$.
We thus consider mean-square stability in the sense of Definition~\ref{def:MSS}. 

The state correlation matrix for system~\eqref{eqn:gen_sys_noise} satisfies the difference equation~\cite{Boyd1994}
\begin{align}
\label{eqn:state_corr_rec_noise}
M(k+1)=\tilde{A}_0M(k)\tilde{A}_0^\mathrm{T} + \tilde{E}_0W\tilde{E}_0^\mathrm{T} + \sum_{i=1}^L\sigma_i^2\left(\tilde{A}_iM(k)\tilde{A}_i^\mathrm{T}+\tilde{E}_iW\tilde{E}_i^\mathrm{T} \right).
\end{align}
The noise covariance $W$ and the matrices $\tilde{E}_0$ and $\tilde{E}_i$ are constant, thus, we introduce $\hat{W} := \tilde{E}_0W\tilde{E}_0^\mathrm{T} + \sum_{i=1}^L\sigma_i^2\tilde{E}_iW\tilde{E}_i^\mathrm{T}$ to simplify notation.
We further define the linear map $\Gamma(X) = \tilde{A}_0X\tilde{A}_0^\mathrm{T} + \sum_{i=1}^L\sigma_i^2\tilde{A}_iX\tilde{A}_i^\mathrm{T}$, then~\eqref{eqn:state_corr_rec_noise} simplifies to
\begin{align}
\label{eqn:state_corr_abbrv}
M(k+1) = \Gamma(M(k)) + \hat{W}.
\end{align}
With that, we can state the following result:
\begin{lem}
\label{lem:MSS_noise}
If the system~\eqref{eqn:gen_sys_noise} with $\epsilon(k)=0\,\forall k$ is MSS according to Definition~\ref{def:MSS_noise_free}, the system defined in~\eqref{eqn:gen_sys_noise} is MSS according to Definition~\ref{def:MSS}.
\end{lem} 
\begin{proof}
We first write~\eqref{eqn:state_corr_abbrv} in explicit form,
\begin{align}
\label{eqn:state_corr_expl}
M(k) = \Gamma^k(M(0)) + \sum_{i=0}^{k-1}\Gamma^i(\hat{W}),
\end{align}
where $\Gamma^k$ denotes the repeated composition of $\Gamma$ with itself and $\Gamma^0$ means identity.
Now taking the limit $k \to \infty$ yields
\begin{align}
\label{eqn:stab_noise_proof}
\begin{split}
\lim_{k\to\infty}M(k) &= \underbrace{\lim_{k\to\infty}\Gamma^k(M(0))}_{= 0\, \text{by Lemma~\ref{lem:LMIcond}}} + \lim_{k\to\infty}\sum_{i=0}^{k-1}\Gamma^i(\hat{W})
 = \lim_{k\to\infty}\sum_{i=0}^{k-1}\Gamma^i(\bar{W}-\Gamma(\bar{W}))\\
 &= \lim_{k\to\infty} \Big( \sum_{i=0}^{k-1}\Gamma^i(\bar{W}) - \sum_{i=1}^{k}\Gamma^i(\bar{W}) \Big)
 = \Gamma^0(\bar{W})-\lim_{k\to\infty}\Gamma^k(\bar{W})
= \bar{W},
\end{split}
\end{align}
where $\bar{W}$ is the solution to $\hat{W}=\bar{W}-\Gamma(\bar{W})$, which exists as the linear map $\Gamma$ is stable (cf.~\cite[p.~132]{Boyd1994}) and was used after the second equal sign. 
The argument $\lim_{k\to\infty}\Gamma^k(M(0)) = 0$ follows as $\Gamma^k(M(0))$ exactly describes the the state correlation matrix in the noise-free case.
As we assume the noise-free system to be MSS, $\lim_{k\to\infty}\Gamma^k(M(0)) = 0$ directly follows from Definition~\ref{def:MSS_noise_free}.
Definition~\ref{def:MSS_noise_free} further requires the state correlation matrix to vanish for any initial $z(0)$ and thus $M(0)$.
Therefore, we can set $M(0)=\bar{W}$ and $\lim_{k\to\infty}\Gamma^k(\bar{W})$ will also vanish.
\end{proof}

That is, by proving mean-square stability of the noise-free system, we also have the stability guarantee for the perturbed system.

We can now use Lemma~\ref{lem:MSS_noise} to prove Theorem~\ref{thm:MSSourSysNoise}. For this we rewrite~\eqref{eqn:matrix_repr} to include noise
\begin{align}
\label{eqn:matr_repr_noise}
\underbrace{\begin{pmatrix}
x(k+1)\\\hat{x}(k+1)\\u(k+1)\\ \hat{u}(k+1)
\end{pmatrix}}_{z(k+1)}
&=\underbrace{\begin{pmatrix}
A&0&B&0\\
\theta A&(1-\theta)A&0&B\\
0&\phi FA&(1-\phi)I&\phi FB\\
0&FA&0&FB
\end{pmatrix}}_{\tilde{A}(k)}
\underbrace{\begin{pmatrix}
x(k)\\ \hat{x}(k)\\u(k)\\ \hat{u}(k)
\end{pmatrix}}_{z(k)} + \underbrace{\begin{pmatrix}
1 & 0\\
0 & \theta\\
0 & 0\\
0 & 0
\end{pmatrix}}_{\tilde{E}(k)}
\underbrace{
\vphantom{\begin{pmatrix}
1 & 0\\
0 & \theta\\
0 & 0\\
0 & 0
\end{pmatrix}}
\begin{pmatrix}
v(k)\\w(k)
\end{pmatrix}}_{\epsilon(k)}.
\end{align}
We employ the same transformation for $\theta$ and $\phi$ as in Theorem~\ref{thm:MSSourSystem}.
The random variables $v(k)$ and $w(k)$ are \iid, zero-mean Gaussian random variables with finite variance ($\Sigma_\mathrm{proc}$ respectively $\Sigma_\mathrm{meas}$).
Thus, all properties of~\eqref{eqn:gen_sys_noise} are satisfied, and Lemma~\ref{lem:MSS_noise} yields the stability result.

\subsection{Proof of Theorem~\ref{thm:switch_system_high_level}}
\label{sec:proof_mode_changes}
To prove Theorem~\ref{thm:switch_system_high_level}, we will employ the following stability result for switched systems:
\begin{lem}[\cite{zhang2008exponential}]
\label{lem:stability_avg_dwell_time}
Consider the discrete-time switched system $x_{k+1} = f_{\sigma(k)}(x_k)$, $\sigma(k)\in\mathcal{F}$ and let $0<\alpha<1$, $\mu>1$ be given constants.
Suppose that there exists a Lyapunov function candidate $V(x) = \{V_{\sigma(k)}(x),\sigma(k)\in\mathcal{F}\}$ satisfying the following properties:
\begin{subequations}
\begin{align}
\label{eqn:cond_a_stab_avg_dwell_time}
\Delta V_{\sigma(k)}\left(x_k\right)&\coloneqq V_{\sigma(k)}\left(x_{k+1}\right)-V_{\sigma(k)} \left(x_k\right) \leq -\alpha V_{\sigma(k)}\left(x_k\right) \quad \forall k\in\left[k_l,k_{l+1}\right]\\
\label{eqn:cond_b_stab_avg_dwell_time}
V_{\sigma(k_l)}\left(x_{k_l}\right) &\leq \mu V_{\sigma(k_{l-1})}\left(x_{k_l}\right).
\end{align}
\label{eqn:cond_stab_avg_dwell_time}
\end{subequations}
Then the system is globally exponentially stable for any switching signal with the average dwell time
\begin{align}
\label{eqn:stab_avg_dwell_time}
\tau_\mathrm{a}\ge\tau_\mathrm{a}^*=\mathrm{ceil}\left[-\frac{\ln \mu}{\ln\left(1-\alpha\right)}\right],
\end{align}
where $\text{ceil}(a)$ is a function rounding $a\in\R$ to the nearest integer greater than or equal to $a$.
\end{lem}

We will now show that this result can be applied to the system defined in~\eqref{eqn:sys_complete}--\eqref{eqn:pred_inp} and guarantees mean-square stability.
To this end, we rewrite~\eqref{eqn:matr_repr_noise} to include switching,
\begin{align}
\label{eqn:matr_repr_noise_switch}
\underbrace{\begin{pmatrix}
x(k+1)\\\hat{x}(k+1)\\u(k+1)\\ \hat{u}(k+1)
\end{pmatrix}}_{z(k+1)}
&=\underbrace{\begin{pmatrix}
A_{\sigma(k)}&0&B_{\sigma(k)}&0\\
\theta A_{\sigma(k)}&(1-\theta)A_{\sigma(k)}&0&B_{\sigma(k)}\\
0&\phi F_{\sigma(k)}A_{\sigma(k)}&(1-\phi)I&\phi F_{\sigma(k)}B_{\sigma(k)}\\
0&F_{\sigma(k)}A_{\sigma(k)}&0&F_{\sigma(k)}B_{\sigma(k)}
\end{pmatrix}}_{\tilde{A}_{\sigma(k)}(k)}
\underbrace{\begin{pmatrix}
x(k)\\ \hat{x}(k)\\u(k)\\ \hat{u}(k)
\end{pmatrix}}_{z(k)} + \underbrace{\begin{pmatrix}
1 & 0\\
0 & \theta\\
0 & 0\\
0 & 0
\end{pmatrix}}_{\tilde{E}(k)}
\underbrace{
\vphantom{\begin{pmatrix}
1 & 0\\
0 & \theta\\
0 & 0\\
0 & 0
\end{pmatrix}}
\begin{pmatrix}
v(k)\\w(k)
\end{pmatrix}}_{\epsilon(k)}.
\end{align}
According to Definitions~\ref{def:MSS} and~\ref{def:MSS_noise_free}, for mean-square stability we care about the evolution of the state correlation matrix.
The state correlation matrix behaves deterministically, as can be seen from~\eqref{eqn:state_corr_rec_noise}. 
We can thus leverage the result from Lemma~\ref{lem:stability_avg_dwell_time}, which is valid for general, deterministic switched systems.

As Theorem~\ref{thm:MSSourSystem} holds, there exists a monotonically decreasing Lyapunov function for every subsystem (\ie realization of $\tilde{A}_{\sigma(k)}(k)$) (cf.~\cite[p.~132]{Boyd1994}).
For a monotonically decreasing function, we can always derive an $\alpha$ that fulfills~\eqref{eqn:cond_a_stab_avg_dwell_time}.
Moreover, as we have a finite number of modes, we can find a $\mu$ fulfilling~\eqref{eqn:cond_b_stab_avg_dwell_time} for all possible switching combinations.
Then we can, by Lemma~\ref{lem:stability_avg_dwell_time}, prove that $\lim_{k\to\infty}M(k)=0$ for the noise-free system, which by Theorem~\ref{thm:MSSourSysNoise} implies stability of the system~\eqref{eqn:matr_repr_noise_switch}. 

\subsection{Controller Implementation}
\label{sec:controller_implementation}

The implementation of the controllers follows the design outlined in Sections~\ref{sec:ctrlDesign} and \ref{sec:sync}.
The system matrices $A$ and $B$ of the cart-pole system are provided by the manufacturer in~\cite{Quanser2012}.
For the stabilization experiments, we design a nominal controller for an update interval of $\Tupdate=\SI{40}{\milli\second}$ via pole placement, and we choose $F$ such that we get closed-loop eigenvalues at \num{0.8}, \num{0.85}, and \num{0.9} (twice).
In experiments with update intervals different from \SI{40}{\milli\second}, we adjust the controller to achieve similar closed-loop behavior.
For the synchronization experiments, we choose $Q_i$ in~\eqref{eqn:cost} for all pendulums as suggested by the manufacturer~\cite{Quanser2012} and set $R_i = 0.1$.
As we here care to synchronize the cart positions, we set the first diagonal entry of $Q_\text{sync}$ to \num{5} and all others to 0.

To derive more accurate estimates of the velocities, filtering can be done at higher update intervals than communication occurs.
For the experiments presented in this paper, estimation and filtering occurs at intervals between \SI{10}{\milli\second} and \SI{20}{\milli\second}, depending on the experiment.


\begin{figure}[t]
\centering
\includegraphics[width=0.7\linewidth]{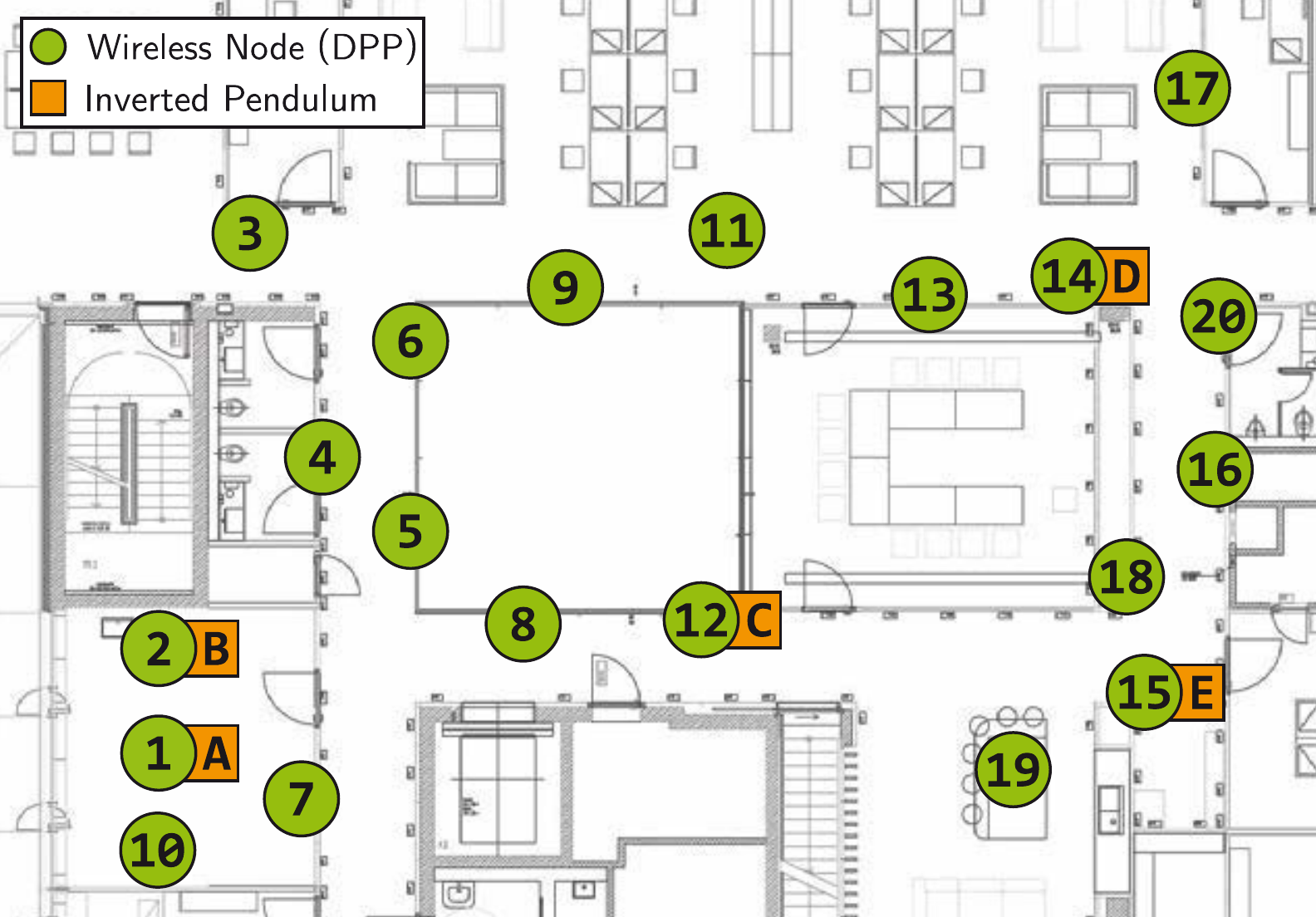}
\caption{Cyber-physical testbed consisting of 20 \dpp nodes that form a three-hop wireless network and five cart-pole systems (two real ones attached to nodes 1 and 2, and three simulated ones at nodes 12, 14, and 15).}
\label{fig:floorplan_mode_changes}
\end{figure}

\end{document}